%% file: ijcai26.tex
\documentclass{article}
\usepackage{ijcai26}

\usepackage{times}
\usepackage{soul}
\usepackage{url}
\usepackage[hidelinks]{hyperref}
\usepackage[utf8]{inputenc}
\usepackage[small]{caption}
\usepackage{graphicx}
\usepackage{amsmath}
\usepackage{amsthm}
\usepackage{amsfonts}
\usepackage{amssymb}
\usepackage{booktabs}
\usepackage{algorithm}
\usepackage{algorithmic}
\usepackage[switch]{lineno}
\usepackage{subcaption}
\usepackage{tikz}
\usepackage[capitalize,noabbrev,nameinlink]{cleveref}
\usepackage{todonotes}
\usepackage{natbib}
\usepackage{comment}
\usepackage{thmtools}
\usepackage{mathtools}
\usepackage{multirow}
\usetikzlibrary{positioning, matrix, shapes.geometric, trees}  

\newtheorem{definition}{Definition}
\newtheorem{proposition}{Proposition}

\DeclareMathOperator*{\argmax}{arg\,max}

\includecomment{appendixcontent} 

\title{Equilibrium Refinements Improve Subgame Solving\\ in Imperfect-Information Games}

\author{
Ondřej Kubíček$^{1, 2}$
\and
Viliam Lisý$^{1}$\And
Tuomas Sandholm$^{2, 3, 4}$\\
\affiliations
$^1$Czech Technical University in Prague\\
$^2$Carnegie Mellon University\\
$^3$Strategy Robot, Inc.\\
$^4$Strategic Machine, Inc.\\
$^5$Optimized Markets, Inc.\\
\emails
kubicon3@fel.cvut.cz, viliam.lisy@agents.fel.cvut.cz,  sandholm@cs.cmu.edu
} 

\begin{document}

\input{constants}
\maketitle
\input{main_body}

\input{chapters/acknowledgments}
\bibliographystyle{named}
\bibliography{ijcai26}
\begin{appendixcontent}
\appendix
\crefalias{section}{appendix}
\include{appendix}
    
\end{appendixcontent}

\end{document}

%% file: constants.tex
\newcommand{\RealNumbers}[0]{\mathbb{R}}
\newcommand{\BinaryNumbers}[0]{\mathbb{B}}
\newcommand{\NaturalNumbers}[0]{\mathbb{N}}
\newcommand{\Expectation}[0]{\mathbb{E}}

\newcommand{\Simplex}{\Delta}

\newcommand{\FOSGame}[0]{\mathcal{G}}
\newcommand{\Subgame}[0]{\FOSGame^{\PublicState}}
\newcommand{\GadgetGame}[1]{\FOSGame^{G,\PublicState}_{#1}}

\newcommand{\PublicIndex}[0]{0}
\newcommand{\InitialIndex}[0]{\text{INIT}}
\newcommand{\BRIndex}[0]{BR}

\newcommand{\Players}[0]{\mathcal{N}}
\newcommand{\Player}[0]{i}
\newcommand{\ChancePlayer}[0]{c}
\newcommand{\OtherPlayer}[0]{j}
\newcommand{\LastPlayer}[0]{N}
\newcommand{\PlayerFunction}[0]{p}

\newcommand{\WorldStates}[0]{\mathcal{W}}
\newcommand{\WorldState}[0]{w}
\newcommand{\InitWorldState}[0]{\WorldState^{\InitialIndex}}

\newcommand{\Histories}[0]{\mathcal{H}}
\newcommand{\History}[0]{h}
\newcommand{\InitHistory}[0]{\History^{\InitialIndex}}
\newcommand{\HistoryExtend}{\sqsubseteq}

\newcommand{\Actions}[1]{\mathcal{A}_{#1}}
\newcommand{\Action}[1]{a_{#1}}
\newcommand{\ActionLogit}[1]{A_{#1}}

\newcommand{\Observations}[1]{\mathcal{O}_{#1}}
\newcommand{\Observation}[1]{o_{#1}}
\newcommand{\ObservationSet}[1]{\mathbb{O}_{#1}}
\newcommand{\PublicObservations}[0]{\Observations{\PublicIndex}}
\newcommand{\PublicObservation}[0]{\Observation{\PublicIndex}}
\newcommand{\PublicObservationSet}[0]{\ObservationSet{\PublicIndex}}

\newcommand{\Rewards}[1]{\mathcal{R}_{#1}}
\newcommand{\Reward}[1]{r_{#1}}
\newcommand{\WeightedReward}[2]{R_{#1}^{#2}}

\newcommand{\Utility}[1]{u_{#1}}

\newcommand{\Transitions}[0]{\mathcal{T}}

\newcommand{\Infosets}[1]{\mathcal{S}_{#1}}
\newcommand{\Infoset}[1]{s_{#1}}
\newcommand{\PublicStates}[0]{\Infosets{\PublicIndex}}
\newcommand{\PublicState}[0]{\Infoset{\PublicIndex}}
\newcommand{\Abstracted}[1]{\overline{#1}}

\newcommand{\AuxiliaryInfoset}[1]{\Infoset{#1}^{G}}
\newcommand{\AuxiliaryInfosets}[1]{\Infosets{#1}^{G}}

\newcommand{\Trajectory}[0]{\tau}
\newcommand{\TrajectoryLength}[0]{l}
\newcommand{\TrajectoryStep}[0]{t}

\newcommand{\TimeStep}[0]{t}
\newcommand{\UnrollStep}[0]{k}

\newcommand{\Strategy}[1]{\pi_{#1}}
\newcommand{\Strategies}[1]{\Pi_{#1}}
\newcommand{\SamplingStrategy}[1]{\mu_{#1}}
\newcommand{\Nash}[1]{\Strategy{#1}^*}
\newcommand{\EpsilonNash}[1]{\Strategy{#1}^{\epsilon}}
\newcommand{\BRStrategy}[1]{\Strategy{#1}^{\BRIndex}}
\newcommand{\BRSet}[1]{\BRIndex_{#1}} 
\newcommand{\BlueprintStrategy}[1]{\overline{\Strategy{#1}}}
\newcommand{\BRBlueprintStrategy}[1]{\overline{\Strategy{#1}^{\BRIndex}}}

\newcommand{\BeforeStrategy}[1]{\hat{\Strategy{#1}}}
\newcommand{\OptimalContinuationStrategy}[1]{\hat{\Strategy{}}_{#1}^*}

\newcommand{\StrategyUtility}[2]{\Utility{#1}^{#2}}
\newcommand{\CounterfactualValue}[2]{v_{#1}^{#2}}
\newcommand{\Regret}[1]{r_{#1}}

\newcommand{\Exploitability}[0]{\mathcal{E}}

\newcommand{\Reach}[1]{P^{#1}}
\newcommand{\PlayerReach}[2]{\Reach{#2}_{#1}}

\newcommand{\StrategyPrior}[1]{\Strategy{#1}^P}

\newcommand{\Sequence}[1]{\sigma_{#1}}
\newcommand{\Sequences}[1]{\Sigma_{#1}}

\newcommand{\SQFStrategy}[1]{\boldsymbol{\Strategy{#1}}}
\newcommand{\SQFValue}[1]{\boldsymbol{v_{#1}}}
\newcommand{\SQFUtility}[0]{\boldsymbol{A}}
\newcommand{\SQFInitRealization}[1]{\boldsymbol{f_{#1}}}
\newcommand{\SQFRealizationConstraints}[1]{\boldsymbol{F_{#1}}}

\newcommand{\SQFPrior}[0]{\boldsymbol{p}}

\newcommand{\PriorEpsilon}[0]{\epsilon}

\newcommand{\PerturbationBasis}[0]{\boldsymbol{B}}

\newcommand{\Transposition}[0]{\top}

%% file: main_body.tex
\input{chapters/abstract}

\input{chapters/1_introduction}
\input{chapters/2_background}
\input{chapters/3_example}
\input{chapters/4_gadget_equilibria}
\input{chapters/5_priors}
\input{chapters/6_resolving_vs_maxmargin}
\input{chapters/7_algorithms}
\input{chapters/8_experiments}
\input{chapters/9_conclusion}

%% file: chapters/abstract.tex
\begin{abstract} 
Subgame solving is a technique for scaling algorithms to large games by locally refining a precomputed blueprint strategy during gameplay. While straightforward in perfect-information games where search starts from the current state, subgame solving in imperfect-information games must account for hidden states and uncertainty about the opponent's past strategy. \textit{Gadget games} were developed to ensure that the improved subgame strategy is robust against any possible opponent's strategy in a zero-sum game. Gadget games typically contain infinitely many Nash equilibria. We demonstrate that while these equilibria are equivalent in the gadget game, they yield vastly different performance in the full game, even when facing a rational opponent. We propose \textit{gadget game sequential equilibria} as the preferred solution concept. We introduce modifications to the sequence-form linear program and counterfactual regret minimization that converge to these refined solutions with only mild additional computational cost. Additionally, we provide several new insights into the surprising superiority of the \textit{resolving gadget game} over the \textit{max-margin gadget game}. Our experiments compare different Nash equilibria of gadget games in several standard benchmark games, showing that our refined equilibria consistently outperform unrefined Nash equilibria, and can reduce the exploitability of the overall strategy by more than 50\%.

\end{abstract}

%% file: chapters/1_introduction.tex
\section{Introduction}
Settings where parties such as people, companies, or nations interact are ubiquitous, and game theory provides the solution concepts of how the parties should act rationally. It is therefore of significant real-world importance to be able to compute (approximate) game-theoretic solutions to large games. Recreational games have often served as computational benchmarks for this purpose. 

In perfect-information games, \textit{real-time subgame solving} techniques have been the core of game-solving algorithms since the beginning of AI.  Algorithms like minimax, alpha-beta search, or Monte Carlo tree search, explore the subgame tree from the single current state of the game forward~\citep{shoham2008multiagent,alphazero2018}. 

However, most real-world settings are imperfect-information games. They present challenges that those algorithms cannot address. There are several real states consistent with a single decision point. Consider poker as an example. The player can see their own cards but not the opponent's. So, each possible opponent hand defines a different state, yet they appear the same from the player's perspective. Furthermore, changing a part of an agent's strategy at some point of a game typically changes what that agent's optimal strategy is later \textit{and earlier} in the game. Thus, a subgame cannot be solved based solely on information from that subgame. 

Therefore, unlike in perfect-information games where subgame solving was a core technique from the outset, imperfect-information games where approached as a monolithic whole for the longest time. Real-time subgame solving started relatively recently~\citep{gilpin2006poker}. Since then, a host of better and better subgame solving techniques have been developed~\citep{gilpin2007abstraction,cfrd2014,ganzfried2015endgame,maxmargin2016,deepstack2017,reachmaxmargin2017,libratus2018,pluribus2019,zhang2021subgame,valuefunctions2023,studentofgames2023,obscuro2025}. In fact, subgame solving has made the central difference in achieving superhuman play in games such as two-player no-limit Texas hold'em poker~\citep{libratus2018}, multi-player no-limit Texas hold'em~\citep{pluribus2019}, and dark chess (aka. Fog-of-War chess)~\citep{obscuro2025}.

A common approach is to compute a coarse \textit{blueprint strategy} for each agent for the entire game, and then to conduct subgame solving on the fly to refine the blueprint strategy.
%
\textit{Gadget games} were developed to compute robust strategies in the subgames. These are modifications of the subgame that allow the opponent to adversarially choose a strategy it could have played in the past. The resulting optimal strategy of the player in such a gadget game is robust against any past strategy of the opponent~\citep{cfrd2014,maxmargin2016,reachmaxmargin2017}. However, existing gadget constructions, such as \textit{resolving} and \textit{max-margin}, only ensure that the player's strategy, in the worst case, is no worse than the blueprint. As a result, if the blueprint is poor, all of these methods may yield poor results.

We make the following contributions. 
1. We show that different equilibria, which all yield the same value in the gadget game, have vastly different performance (exploitability) in the full game, when combined with the blueprint. We introduce a weaker notion of sequential equilibrium, \textit{gadget game sequential equilibrium}, which is easier to compute, and propose that gadget games should be solved according to this concept. 2. We disprove the widely held belief that max-margin gadget games are strictly superior to resolving gadget games. This provides a new understanding of some prior works that have used the resolving gadget game for subgame solving rather than max-margin~\citep{deepstack2017,obscuro2025}. 3. We introduce efficient modifications to the \textit{sequence-form linear program (SQF)} and \textit{counterfactual regret minimization (CFR)} that converge to these refined equilibria with minimal additional computational overhead. 4. We empirically evaluate different types of Nash equilibria in several standard benchmark games. The results show that converging to gadget game sequential equilibria can reduce exploitability by more than 50\% compared to vanilla SQF and CFR. Moreover, the results show that in resolving gadget games, these equilibria exhibit similar performance gains as using unsafe subgame solving, but provably without the risk of performing worse than blueprint. 


%% file: chapters/2_background.tex
\section{Background and Notation}
Two-player zero-sum factored-observation stochastic game is $\FOSGame = (\Players,\WorldStates, \InitWorldState, \PlayerFunction, \Actions{}, \Transitions, \Rewards{1},  \Observations{})$, where $\Players = \{1, 2, c\}$ is a set of players and $\ChancePlayer$ is a chance player, which plays fixed strategy and it is used to model stochasticity of the environment. $\WorldStates$ is the set of world states and $\InitWorldState \in \WorldStates$ is an initial world state. $\PlayerFunction: \WorldStates \to 2^\Players$ is a function that assigns acting players to a given world state $\WorldState$. $\Actions{} = \prod_{\Player \in \Players} \Actions{\Player}$ is the set of joint actions, we also use $\Actions{\Player}(\WorldState) \subseteq \Actions{\Player}$ to denote legal actions of $\Player$ in state $\WorldState$. $\Transitions:\WorldStates \times \Actions{} \to \WorldStates$ is a transition function and $\Rewards{1}: \WorldStates \times \Actions{} \to \RealNumbers $ is a reward function of player $1$. In a zero-sum setting following holds $\Rewards{2}(\WorldState, \Action{}) = -\Rewards{1}(\WorldState, \Action{})$. $\Observations{}: \WorldStates \times \Actions{} \times \WorldStates \to \ObservationSet{}$ is the observation function and $\ObservationSet{}: \ObservationSet{\PublicIndex} \times \ObservationSet{1} \times \ObservationSet{2}$ is the set of joint public observations and private observations. Similarly, observation function can be factored as $\Observations{} = (\Observations{\PublicIndex}, \Observations{1}, \Observations{2})$, where $\Observations{\PublicIndex}$ is a public observation given to each player and $\Observations{1}, \Observations{2}$ are private observations available to their respective players \citep{kovarik2022fosg}.

History $\History = \WorldState^0 \Action{}^0 \dots \Action{}^{\TrajectoryLength-1} \WorldState^{\TrajectoryLength} \in (\WorldStates \Actions{})^* \WorldStates$ is a finite sequence of world states and actions, which starts in the initial state $\WorldState^0 = \InitWorldState$ and each timestep $\TimeStep \in \{0, \dots, \TrajectoryLength - 1\}$ is a valid transition $\WorldState^{\TimeStep +1} = \Transitions(\WorldState^{\TimeStep}, \Action{}^{\TimeStep})$. $\Histories$ is a set of all possible histories in a game $\FOSGame$. We use $\History \HistoryExtend \History'$ to denote $\History'$ that contains $\History$ as a prefix and we will denote the initial history as $\InitHistory := \InitWorldState$. Every history ends with some world state $\WorldState^{\TrajectoryLength}$, we will often use history in the game functions instead of that world state, for example $\Actions{}(\History) := \Actions{}(\WorldState^{\TrajectoryLength})$. The players do not observe the whole world state at each timestep, but they only observe public and their corresponding private observations. As a result, the player may not distinguish several different histories. $\Infoset{\Player}$ is an information set of the player, which is a set of all histories consistent with the observations of player $\Player$. $\Infosets{\Player}$ is then set of all information sets. We will overload the notation and use $\Infoset{\Player}(\History)$ to denote the information set corresponding to history $\History$ and $\Histories(\Infoset{\Player})$ as all histories consistent with information set $\Infoset{\Player}$. Similarly, $\PublicState \in \PublicStates$ is a public state, which is an information set of the external player, that does not have any private observations, so it contains all the histories consistent with public observations. As a result, each public state contains one or more information sets of each player and we will use $\PublicState(\Infoset{\Player})$ to denote the public state which contains information set $\Infoset{\Player}$ and $\Infosets{\Player}(\PublicState)$ as a set of all information sets consistent with $\PublicState$.

The behavioral strategy of the player $i$ is $\Strategy{\Player} : \Infosets{\Player} \to \Simplex \Actions{\Player}$ a mapping from information sets to probability distribution over actions. We will often use $\Strategy{\Player}(\Infoset{\Player}, \Action{\Player})$ as a probability that $\Action{\Player}$ is played in $\Infoset{\Player}$, when following $\Strategy{\Player}$. The strategy profile $\Strategy{} = (\Strategy{1}, \Strategy{2})$ is a joint strategy profile of both players. For any two histories $\History, \History'$, such that $\History \HistoryExtend \History'$, the probability of reaching $\History'$ from $\History$ under strategy profile $\Strategy{}$ is $\Reach{\Strategy{}}(\History' | \History) = \prod_{\History'' \Action{} \WorldState \HistoryExtend \History'} \prod_{\Player \in \Players} \Strategy{\Player}(\Infoset{\Player}(\History''), \Action{\Player})$. Any reach probability can be factored into the components of the individual players $\Reach{\Strategy{}}(\History' | \History) = \prod_{\Player \in \Players} \PlayerReach{\Player}{\Strategy{}}(\History' | \History)$. Sometimes we use $\Reach{\Strategy{}}(\History) := \Reach{\Strategy{}}(\History | \InitHistory)$. 

The expected utility of history $\History$ if all players follow strategy profile $\Strategy{}$ is $\Utility{\Player}^{\Strategy{}}(\History) = \sum_{\History \HistoryExtend \History' \Action{}} \Reach{\Strategy{}}(\History'|\History)\Rewards{\Player}(\History', \Action{}) \prod_{\Player \in \Players} \Strategy{\Player}(\Infoset{\Player}(\History'), \Action{\Player})$. Moreover, $\CounterfactualValue{\Player}{\Strategy{}}(\Infoset{\OtherPlayer}) = \frac{\sum_{\History \in \Histories(\Infoset{\OtherPlayer})} \PlayerReach{-\OtherPlayer}{\Strategy{}}(\History) \Utility{\Player}^{\Strategy{}}(\History)}{\sum_{\History \in \Histories(\Infoset{\OtherPlayer})} \PlayerReach{-\OtherPlayer}{\Strategy{}}(\History)}$ are a counterfactual values. If $\sum_{\History \in \Histories(\Infoset{\OtherPlayer})} \PlayerReach{-\OtherPlayer}{\Strategy{}}(\History) = 0$, then $\CounterfactualValue{\Player}{\Strategy{}}(\Infoset{\OtherPlayer}) = 0$. The best response against a strategy $\Strategy{\Player}$ is a strategy $\BRStrategy{-\Player} \in \BRSet{-\Player}$ that maximizes the opponents utility $\BRStrategy{-\Player} = \argmax_{\Strategy{-\Player}} \Utility{-\Player}^{(\Strategy{\Player}, \Strategy{-\Player})}(\InitHistory)$. The subscript $-\Player$ is used to identify the other player then $\Player$. When each player follows a best response strategy to each other, the resulting strategy profile is the Nash equilibrium $\Nash{}$ \citep{nash1950}.  In two-player zero-sum games, the Nash equilibrium is the desired solution concept and we use exploitability as a metric to evaluate quality of a strategy $\Exploitability(\Strategy{\Player}) = \Utility{-\Player}^{\Strategy{\Player}, \BRStrategy{-\Player}} (\InitHistory) - \Utility{-\Player}^{\Nash{}}(\InitHistory)$, which is how much can the opponent gain if it plays a best response compared to the value received by both players following Nash equilibrium. The exploitability is always non-negative and is zero if and only if the $\Strategy{\Player}$ is a part of Nash equilibrium.

\subsection{Subgame Solving}

There are two main approaches to subgame solving in imperfect-information games. The first, common-knowledge subgame solving, considers all histories that share the public state~\citep{cfrd2014,deepstack2017,libratus2018,studentofgames2023,sepot2024}. The second, knowledge-limited subgame solving, considers only a subset of public state histories~\citep{zhang2021subgame,liu2023opponent,obscuro2025}. Although our work focuses on common-knowledge subgame solving, the principles apply to both paradigms.

A subgame $\Subgame$ of a game $\FOSGame$ is defined the same as the original game, but it has different initial world state $\WorldState^{\InitialIndex, \PublicState}$ in which the chance player chooses between all the histories $\History \in \Histories(\PublicState)$, where $\PublicState$ is the public state which contains current decision. Assume all players followed a precomputed blueprint strategy $\BlueprintStrategy{}$, the chance distribution in the initial state of the subgame is proportional to $\Reach{\BlueprintStrategy{}}(\History)$. This results in a valid game that can be solved with any algorithm for imperfect-information games. 

Solving such a subgame $\Subgame$ results in a strategy profile $\Strategy{}^{\Subgame}$. We denote combined strategy as $\BlueprintStrategy{\Player} \gets \Strategy{\Player}^{\Subgame}$, where player $\Player$ follows the new strategy $\Strategy{\Player}^{\Subgame}$ within the subgame and the blueprint elsewhere $\BlueprintStrategy{}$ . This approach is known as unsafe subgame solving. While it often performs suprisingly well in practice~\citep{gilpin2006poker,gilpin2007abstraction,ganzfried2015endgame,reachmaxmargin2017}, it lacks theoretical guarantees~\citep{cfrd2014,ganzfried2015endgame,valuefunctions2023}. The unsafety stems from assuming the opponent followed a fixed strategy in the past. However, if the opponent has played according to a different strategy, the game will look exactly the same, but the computed solution would perform vastly differently.

To rectify this, gadget games were developed to compute robust strategies \citep{cfrd2014,maxmargin2016,reachmaxmargin2017}. The gadget game $\GadgetGame{\Player}$ is similar to  $\Subgame$, but instead of assuming the opponent followed the blueprint $\BlueprintStrategy{}$, the gadget game introduces auxiliary information sets $\AuxiliaryInfosets{}$ that allow the opponent to choose their reach probabilities adversarially.

\begin{definition}
    Let $\Subgame$ be a subgame and $\GadgetGame{\Player}$ be a gadget game of player $\Player$ corresponding to $\Subgame$. \emph{Auxiliary information sets} $\AuxiliaryInfosets{}$ are all information sets that are in $\GadgetGame{\Player}$ but not in $\Subgame$.
\end{definition}

Different types of gadget games introduce different auxiliary information sets $\AuxiliaryInfosets{}$. The \textit{resolving gadget game} contains an auxiliary information set for each opponent's information set $\Infoset{-\Player} \in \Infosets{-\Player}(\PublicState)$ in the root of the subgame, while the \textit{max-margin gadget game} contains only a single auxiliary information set. In \cref{fig:seq} we show the construction of resolving and max-margin gadget games for a specific game. The \textit{reach max-margin gadget game} has the same construction as max-margin, but uses different utility shifts~\citep{reachmaxmargin2017}.

The strategy resulting after solving the gadget game $\BlueprintStrategy{} \gets \Strategy{\Player}^{\GadgetGame{\Player}}$ is guaranteed to be safe, which means it does not increase the exploitability compared to that of the agent's blueprint strategy: $\Exploitability(\BlueprintStrategy{} \gets \Strategy{\Player}^{\GadgetGame{\Player}}) \leq \Exploitability(\BlueprintStrategy{})$.

%% file: chapters/3_example.tex
\section{Example}

\begin{figure*}
    \centering    
    \begin{subfigure}[b]{0.32\textwidth}
        \centering
        \input{tikz/sequential_example}
        \caption{Original game}
        \label{fig:seq_ex}
    \end{subfigure}
    \hfill 
    \begin{subfigure}[b]{0.32\textwidth}
        \centering
        \input{tikz/sequential_resolving}
        \caption{Resolving gadget game}
        \label{fig:seq_resolve}
    \end{subfigure}
    \hfill 
    \begin{subfigure}[b]{0.32\textwidth}
        \centering
        \input{tikz/sequential_maxmargin}
        \caption{Max-margin gadget game}
        \label{fig:seq_maxmargin}
    \end{subfigure}
    \caption{Example game, that highlights the problem with sequential rationality in subgame solving. Player 1 (red) is maximizing the value, while Player 2 (blue) is minimizing. Chance player (grey) is playing according to a shown fixed "probability" distribution.}
    \label{fig:seq} 
\end{figure*}
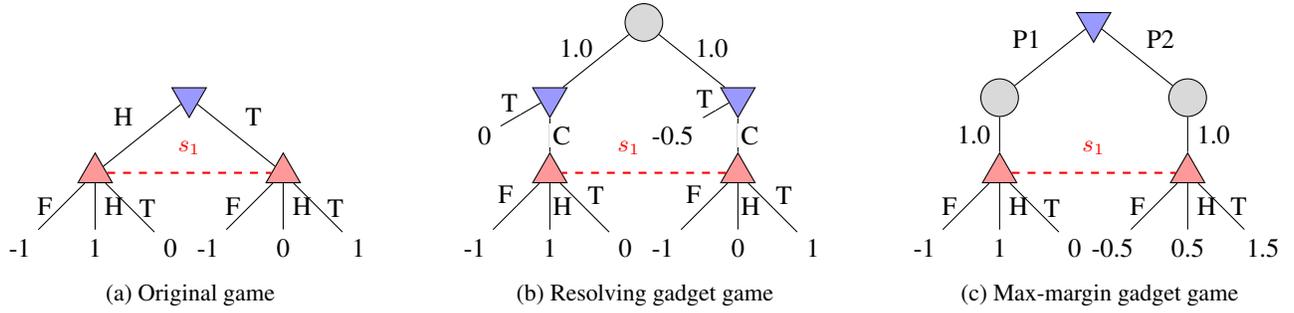

To demonstrate the issues with Nash equilibria in gadget games, we present a modified version of Matching Pennies in \cref{fig:seq_ex}. The game begins with Player 2 choosing either heads (H) or tails (T). Player 1, unaware of this choice, must choose among three actions: Forfeit (F), Heads (H), and Tails (T). By forfeiting, player 1 loses with a reward of -1. Otherwise, Player 1 receives a reward of 1 if they match Player 2's action and 0 if they do not. This game has a unique Nash equilibrium in which both players play H and T with probability 0.5.

Assume a blueprint strategy for player 1, where $\BlueprintStrategy{1} = [0.5, 0.5, 0]$ for actions F, H, T, respectively. Under this blueprint, the counterfactual values of the best response in the left and right branches of the tree are $\CounterfactualValue{1}{\BlueprintStrategy{1}}(H) = 0$, $\CounterfactualValue{1}{\BlueprintStrategy{1}}(T) =-0.5$.

In a resolving gadget game, Player 2 has the option to either enter the subgame (C) or terminate and receive the counterfactual best response value to the blueprint (T). The initial state of the subgame is decided by a chance node with probabilities proportional to the reaches $\PlayerReach{1}{\BlueprintStrategy{1}} \cdot \PlayerReach{c}{}$. These ``probabilities'' do not necessarily sum to 1. \cref{fig:seq_resolve} illustrates the extensive-form representation of this gadget game. Any strategy of Player 1 that forces Player 2 to play into a terminating action is a Nash equilibrium. Crucially, this includes even the blueprint, which plays suboptimal action F, which is strictly dominated.

The max-margin gadget game allows Player 2 first to select an information set before the chance node determines the specific state based on the reaches $\PlayerReach{1}{\BlueprintStrategy{1}} \cdot \PlayerReach{c}{}$. This construction also subtracts the counterfactual best response values from each terminal utility. \cref{fig:seq_maxmargin} shows the extensive-form representation of this game. In this specific case, there is a unique Nash equilibrium $\Strategy{1} = [0, 0.75, 0.25]$, which is not part of a Nash equilibrium of the whole game.

This example already hints at several of our results, which we discuss in more detail with more complicated examples in \cref{app:proofs}. First, gadget games can have infinitely many equilibria that yield different values in the original game (\cref{thm:nonoptimality}). Second, with some blueprint strategy, one might not be able to construct the optimal \textit{continuation strategy} (\cref{thm:nonexistence}). Third, some equilibria in resolving gadget games can outperform all Nash equilibria of a max-margin gadget game (\cref{thm:maxmarginworse}).

%% file: tikz/sequential_example.tex
\begin{tikzpicture}[
    level 1/.style={sibling distance=2.5cm, level distance=1cm},
    level 2/.style={sibling distance=1cm, level distance=1cm},
    p2node/.style={ 
      regular polygon,
      regular polygon sides=3,
      rotate=180, 
      draw,
      fill=blue!40,
      inner sep=1pt,
      minimum size=15pt
    },
    p1node/.style={
      regular polygon,
      regular polygon sides=3, 
      draw,
      fill=red!40,
      inner sep=1pt,
      minimum size=15pt
    },
    payoff/.style={
      font=\normalsize
    }
  ]

  \node [p2node] (root) {}
    child {
      node [p1node] (p1_left) {} 
      child { node [payoff] {-1} edge from parent node [left] {F} }
      child { node [payoff] {1}  edge from parent node [right] {H} } 
      child { node [payoff] {0} edge from parent node [right] {T} }
      edge from parent node [above left] {H} 
    }
    child {
      node [p1node] (p1_right) {} 
      child { node [payoff] {-1}  edge from parent node [left] {F} }
      child { node [payoff] {0}  edge from parent node [right] {H} } 
      child { node [payoff] {1} edge from parent node [right] {T} }
      edge from parent node [above right] {T} 
    };

  \draw [red, dashed, thick] (p1_left) -- (p1_right)
        node [midway, above, yshift=2mm, fill=white, inner sep=1pt, font=\small] {$\Infoset{1}$};

\end{tikzpicture}

%% file: tikz/sequential_resolving.tex
\begin{tikzpicture}[
    level 1/.style={sibling distance=2.5cm, level distance=1cm},
    level 2/.style={sibling distance=1cm, level distance=1cm},
    level 3/.style={sibling distance=1cm, level distance=1cm},
    p2node/.style={ 
      regular polygon,
      regular polygon sides=3,
      rotate=180, 
      draw,
      fill=blue!40,
      inner sep=1pt,
      minimum size=15pt
    },
    p1node/.style={
      regular polygon,
      regular polygon sides=3, 
      draw,
      fill=red!40,
      inner sep=1pt,
      minimum size=15pt
    },
    chancenode/.style={circle, draw, fill=gray!30, minimum size=5mm, inner sep=1pt},
    terminal/.style={inner sep=2pt},
    payoff/.style={
      font=\normalsize
    },
    edge_label/.style={midway, fill=white, inner sep=1pt}
]

\node[chancenode] (S) {} 
    child {
        node[p2node] (P2-L) {}
        child[grow=-150] { 
            node[payoff] (T1) {0}
            edge from parent
            node[edge_label, above left] {T}
        }
        child[grow=-90] { 
            node[p1node] (P1-L) {}
            child {
                node[payoff] (T2) {-1}
                edge from parent
                node[edge_label, above left] {F}
            }
            child {
                node[payoff] (T3) {1}
                edge from parent
                node[edge_label, right] {H}
            }
            child {
                node[payoff] (T4) {0}
                edge from parent
                node[edge_label, above right] {T}
            }
            edge from parent
            node[edge_label, right] {C}
        }
        edge from parent
        node[edge_label, above left] {1.0}
    }
    child {
        node[p2node] (P2-R) {}
        child[grow=-150] { 
            node[payoff] (T5) {-0.5}
            edge from parent
            node[edge_label, above left] {T}
        }
        child[grow=-90] { 
            node[p1node] (P1-R) {}
            child {
                node[payoff] (T6) {-1}
                edge from parent
                node[edge_label, above left] {F}
            }
            child {
                node[payoff] (T7) {0}
                edge from parent
                node[edge_label, right] {H}
            }
            child {
                node[payoff] (T8) {1}
                edge from parent
                node[edge_label, above right] {T}
            }
            edge from parent
            node[edge_label, right] {C}
        }
        edge from parent
        node[edge_label, above right] {1.0}
    };
  \draw [red, dashed, thick] (P1-R) -- (P1-L)
        node [midway, above left, yshift=2mm, fill=white, inner sep=1pt, font=\small] {$\Infoset{1}$};

\end{tikzpicture}

%% file: tikz/sequential_maxmargin.tex
\begin{tikzpicture}[
    level 1/.style={sibling distance=2.5cm, level distance=1cm},
    level 2/.style={sibling distance=1cm, level distance=1cm},
    p2node/.style={ 
      regular polygon,
      regular polygon sides=3,
      rotate=180, 
      draw,
      fill=blue!40,
      inner sep=1pt,
      minimum size=15pt
    },
    p1node/.style={
      regular polygon,
      regular polygon sides=3, 
      draw,
      fill=red!40,
      inner sep=1pt,
      minimum size=15pt
    },
    chancenode/.style={circle, draw, fill=gray!30, minimum size=5mm, inner sep=1pt},
    payoff/.style={
      font=\normalsize
    }
  ]

  \node [p2node] (root) {}
    child {
        node [chancenode] (chance_left) {}
        child {
          node [p1node] (p1_left) {} 
          child { node [payoff] {-1} edge from parent node [left] {F} }
          child { node [payoff] {1}  edge from parent node [right] {H} } 
          child { node [payoff] {0} edge from parent node [right] {T} }
        edge from parent node [left] {1.0} 
      }
      edge from parent node [above left] {P1} 
    }
    child {
        node [chancenode] (chance_right) {}
        child {
          node [p1node] (p1_right) {} 
          child { node [payoff] {-0.5}  edge from parent node [left] {F} }
          child { node [payoff] {0.5}  edge from parent node [right] {H} } 
          child { node [payoff] {1.5} edge from parent node [right] {T} }
        edge from parent node [right] {1.0} 
        }
      edge from parent node [above right] {P2} 
    };

  \draw [red, dashed, thick] (p1_left) -- (p1_right)
        node [midway, above, yshift=2mm, fill=white, inner sep=1pt, font=\small] {$\Infoset{1}$};

\end{tikzpicture}

%% file: chapters/4_gadget_equilibria.tex
\section{Gadget Game Sequential Equilibria}

The goal of subgame solving is to find a local strategy $\Strategy{\Player}^{\Subgame}$ that ensures the updated strategy improves upon the blueprint strategy $\BlueprintStrategy{\Player}$. This improvement is formally characterized in \cref{thm:gadgetvalue} (inspired by \cite{deepstack2017}), which shows that the resulting exploitability depends on the counterfactual best response values in the root information sets of the opponent. 

\input{theorems/gadget_value}

The proof is in \cref{app:proofs}. Any Nash equilibrium of the resolving gadget game ensures that exploitability cannot increase by satisfying $\CounterfactualValue{\Player}{(\Strategy{\Player}, \BRStrategy{-\Player})}(\Infoset{-\Player}) - \CounterfactualValue{\Player}{(\BlueprintStrategy{\Player}, \BRBlueprintStrategy{-\Player})}(\Infoset{-\Player})\geq 0$. In a max-margin gadget game, any Nash equilibrium maximizes the smallest improvement in these counterfactual values.
\begin{equation*}
    \max_{\Strategy{\Player}^{\Subgame{}}} \min_{\Infoset{-\Player}' \in \Infosets{-\Player}}{\CounterfactualValue{\Player}{(\BlueprintStrategy{\Player} \gets \Strategy{\Player}^{\Subgame{}}, \BRStrategy{-\Player})}(\Infoset{-\Player}') - \CounterfactualValue{\Player}{(\BlueprintStrategy{\Player}, \BRBlueprintStrategy{-\Player})}(\Infoset{-\Player}')}
\end{equation*}

Because this minimal improvement is always at least zero, every solution of the max-margin gadget game is also a solution of the corresponding resolving gadget game.

\input{theorems/resolving_prop}

The max-margin gadget game can be further improved by allowing the player to pay back the ``gifts'' that the opponent has given through mistakes so far~\citep{reachmaxmargin2017}. However, this typically requires knowledge outside of the current subgame. We demonstrate that strategies can be improved without additional information by accounting for sequential rationality.

\input{theorems/gadget_equilibria}

We provide an example for this and all future observations in \cref{app:proofs}. In two-player zero-sum games, all Nash equilibrium yield  the same unique value of the game \citep{gtessentials1952}. This is true even in a gadget game. However, thanks to \cref{thm:nonoptimality}, we see that in the gadget game, not all equilibria are equal, as some lead to lower exploitability when used in the whole game. This is because the exploitability of the global strategy does not depend on the value in the root of the gadget game $\GadgetGame{}$, but on the values in the information sets where the subgame $\Subgame$ begins.

Some Nash equilibria may play irrationally in parts of the gadget game that have zero reach. However, those same parts may be reached with high probability in the original game even if the opponent is playing an equilibrium of the whole game. Therefore, ensuring rational play even in unreachable nodes would be an improvement.

There has been a long line of work studying different Nash equilibria refinements in extensive-form games that ensure rationality in unreachable parts of the game \citep{kreps1982sequential,bielefeld1988reexamination,van1984relation,van2012refinements}. The common problem with many of those refinements is that their computation necessitates playing each action with a small non-zero probability, and that the mistakes by the players should not be correlated. These conditions often make the computation numerically unstable. We observe that for gadget games, we do not need the full complexity of these refinements. Since the discrepancy between the gadget and the original game is in the auxiliary information sets, we only need to ensure rationality at the transition from the gadget game to the subgame.

Both the max-margin and resolving gadget game only add a single auxiliary information set to each trajectory, so it is enough to force at most one decision by the opponent to be played with non-zero probability in each trajectory. This means that the desired solution concept in gadget games is a weaker notion than some other refinements studied in the past~\citep{van2012refinements}, which simplifies their computation.

\begin{definition}
\label{def:prior}
    Let $\AuxiliaryInfosets{-\Player}$ be the opponent's auxiliary information sets in the gadget game $\GadgetGame{\Player}$, and let $\StrategyPrior{}(\AuxiliaryInfoset{-\Player}, \Action{-\Player}) > 0$ be a prior strategy $\AuxiliaryInfoset{-\Player} \in \AuxiliaryInfosets{-\Player}$. \emph{Gadget game sequential equilibrium (GGSE)} is a Nash equilibrium of a Gadget game $\GadgetGame{\Player}$, where every action $\Action{-\Player} \in \Actions{-\Player}(\AuxiliaryInfoset{-\Player})$ is played with a probability greater than $\epsilon \cdot \StrategyPrior{}(\AuxiliaryInfoset{-\Player}, \Action{-\Player})$ as $\epsilon \to 0^{+}$. 
\end{definition}

Gadget game sequential equilibria are a subset of Nash equilibria. They play rationally also in information sets to which the opponent plays with probability 0 from the auxiliary information sets.

\input{theorems/ggse_superiority}

Thanks to \cref{thm:ggse}, it is enough to search the space of gadget game sequential equilibria without the risk of excluding equilibria that would result in lower exploitability.

%% file: theorems/gadget_value.tex
\begin{restatable}[]{theorem}{gadgetvalue}
\label{thm:gadgetvalue}
Given a blueprint strategy $\BlueprintStrategy{\Player}$, a subgame $\Subgame{}$, and a subgame strategy $\Strategy{\Player}^{\Subgame{}}$, let $\Strategy{\Player} = \BlueprintStrategy{\Player} \gets \Strategy{\Player}^{\Subgame{}}$. Let $\BRBlueprintStrategy{-\Player} \in \BRSet{-\Player}(\BlueprintStrategy{\Player})$ and $\BRStrategy{-\Player} \in \BRSet{-\Player}(\Strategy{\Player})$. If $\CounterfactualValue{\Player}{(\Strategy{\Player}, \BRStrategy{-\Player})}(\Infoset{-\Player}) \geq \CounterfactualValue{\Player}{(\BlueprintStrategy{\Player}, \BRBlueprintStrategy{-\Player})}(\Infoset{-\Player})$ for each $\Infoset{-\Player} \in \Infosets{-\Player}(\PublicState)$, then $\StrategyUtility{\Player}{(\Strategy{\Player}, \BRStrategy{-\Player})}(\InitHistory) - \StrategyUtility{\Player}{(\BlueprintStrategy{\Player}, \BRBlueprintStrategy{-\Player})}(\InitHistory) \geq \sum_{\Infoset{-\Player} \in \Infosets{-\Player}(\PublicState)} \PlayerReach{-\Player}{\BRStrategy{-\Player}}(\Infoset{-\Player}) \big(\CounterfactualValue{\Player}{(\Strategy{\Player}, \BRStrategy{-\Player})}(\Infoset{-\Player}) - \CounterfactualValue{\Player}{(\BlueprintStrategy{\Player}, \BRBlueprintStrategy{-\Player})}(\Infoset{-\Player})\big) $ 
\end{restatable}

%% file: theorems/resolving_prop.tex
\begin{proposition}
    Every Nash equilibrium of max-margin gadget game is also a Nash equilibrium of the corresponding resolving gadget game.
\end{proposition}

%% file: theorems/gadget_equilibria.tex
\begin{restatable}[]{observation}{gadgetequilibria}
\label{thm:nonoptimality}
    There exists a game $\FOSGame$, a subgame $\Subgame{}$ of $\FOSGame$,  and blueprint $\BlueprintStrategy{\Player}$ of player $\Player$, with corresponding resolving or maxmargin gadget game $\GadgetGame{}$, in which different Nash equilibria result in different exploitability in the original game. 
\end{restatable}

%% file: theorems/ggse_superiority.tex
\begin{restatable}[]{theorem}{ggse}
\label{thm:ggse}
    Assume a subgame $\Subgame$ and the corresponding gadget game $\GadgetGame{\Player}$. Let $\BlueprintStrategy{}$ be a blueprint strategy. For every Nash equilibrium $\Nash{}$ in the gadget game $\GadgetGame{\Player}$, there exists a gadget game sequential equilibrium $\Strategy{}^G$, that does not have lower counterfactual best response values in each root information set $\Infoset{-\Player} \in \Infosets{-\Player}$:
    \begin{align*}
        \min_{ \Strategy{-\Player}} \CounterfactualValue{\Player}{\BlueprintStrategy{\Player} \gets \Nash{\Player}, \Strategy{-\Player}}(\Infoset{-\Player}) \leq \min_{ \Strategy{-\Player}} \CounterfactualValue{\Player}{\BlueprintStrategy{\Player} \gets \Strategy{\Player}^G, \Strategy{-\Player}}(\Infoset{-\Player})
    \end{align*} 
    
\end{restatable}

%% file: chapters/5_priors.tex
\section{Choice of Prior}
While GGSE ensures rationality in the subgame, the performance of the resulting strategy is influenced by the choice of the prior $\StrategyPrior{}$. We first define the ideal target, which is the optimal strategy player $\Player$ could play in the subgame, without changing strategy in the past.

\begin{definition}
\label{def:optimality}
    Let $\Subgame$ be a subgame of game $\FOSGame$. Also, let $\BeforeStrategy{\Player}$ be a blueprint strategy, but only in the information sets $\Infosets{\Player}'$, which preceded some information set from $\Infosets{\Player}(\PublicState)$. The \emph{optimal continuation strategy} $\OptimalContinuationStrategy{\Player}$ satisfies
    \begin{equation}
        \OptimalContinuationStrategy{\Player} = \argmax_{\Strategy{\Player}} \min_{\Strategy{-\Player}} \Utility{\Player}^{(\Strategy{\Player} \gets \BeforeStrategy{\Player}, \Strategy{-\Player} )}(\InitHistory)
    \end{equation}
\end{definition}

There could be multiple optimal continuation strategies, and ideally, we would like to find a strategy that is part of some optimal continuation strategy $\OptimalContinuationStrategy{\Player}$. Finding such a strategy naively requires solving the whole game with a fixed strategy of player $\Player$ before reaching the subgame. A more scalable approach is to use a trained value function outside of the path to the subgame, which is called ``full gadget'' by \citet{milec2024continual}. However, this approach requires a value function, which is difficult to train and grows linearly with the length of the game, making it less tractable than the resolving or max-margin gadget games. 

\input{theorems/non_existence}

\cref{thm:nonexistence} demonstrates that when using the resolving or max-margin gadget game, the optimal continuation strategy may not be reconstructible. 

In practice, many prior works have used unsafe resolving, which assumes the opponent followed the blueprint in the past~\citep{gilpin2006poker,gilpin2007abstraction,ganzfried2015endgame,reachmaxmargin2017}. Although it often yields stronger performance in practice than prior safe subgame solving techniques, it runs the risk of producing strategies that are more exploitable than the blueprint. We propose using the blueprint as the prior, which combines the best of both worlds, as the strategy will always be safe, but it will leverage the information from the blueprint. 

Specifically, for a resolving Player 1, we set the prior of auxiliary information set $\AuxiliaryInfoset{2}$ and action $\Action{}$, which lead to information set $\Infoset{2}$ in the root of the subgame to be proportional to the reach of $\Infoset{2}$.
\begin{equation*}
    \StrategyPrior{2}(\AuxiliaryInfoset{2}, \Action{2}) \propto \PlayerReach{2}{\BlueprintStrategy{2}}(\Infoset{2})
\end{equation*}

The blueprint may play some actions with probability 0, so in our experiments, we clip the value to $\StrategyPrior{2}(\AuxiliaryInfoset{2}, \Action{2}) \geq  10^{-3}$. 

Using the opponent's part of the blueprint in subgame solving was explored in the poker AI Deepstack, where it was used as a warm-start strategy in the root of the gadget game \citep{deepstack2017}. Contrary to our approach, they used a blueprint only in the first iteration, so their approach did not generally converge to GGSE.

%% file: theorems/non_existence.tex
\begin{restatable}[]{observation}{nonexistence}
\label{thm:nonexistence} 
    There exists a game $\FOSGame$, a subgame $\Subgame{}$ of $\FOSGame$, and blueprint strategy $\BlueprintStrategy{}$ of player $\Player$, with corresponding resolving or maxmargin gadget game $\GadgetGame{}$, in which no Nash equilibrium is part of the optimal continuation strategy $\OptimalContinuationStrategy{\Player}$
\end{restatable}

%% file: chapters/6_resolving_vs_maxmargin.tex
\section{Superiority of the Resolving Gadget Game}
There has been a long-standing consensus that the max-margin gadget game is an improvement over the resolving gadget game. However, the authors of the AI poker system Deepstack noted that resolving gadget game performed better in early experiments \citep{deepstack2017}. Similarly, the Obscuro AI agent for dark chess solves both the resolving and max-margin gadget games, and then heuristically chooses one of those two~\citep{obscuro2025}.

Previously, these discrepancies were often attributed to the use of imperfect value functions in nested subgame solving. The logic was that if the value function is inaccurate, the safety constraints from \cref{thm:gadgetvalue} might be impossible to satisfy, leading to more exploitable strategies.

We offer another explanation that shows the max-margin gadget game can be outperformed by resolving even in endgame situations or when using a perfect value function. We show that the max-margin gadget game with a poor blueprint may eliminate some equilibria that are still obtainable by the resolving gadget game. Moreover, the max-margin gadget game may remove all equilibria that are part of the optimal continuation strategy $\OptimalContinuationStrategy{\Player}$.

\input{theorems/maxmargin_worse}

The reason for this performance degradation is that in the max-margin gadget game, the opponent chooses at the root which information set it should play into based on how much improvement the player can get there, which introduces a bias to any Nash equilibrium strategy. In contrast, the resolving gadget game does not exhibit this bias toward certain parts of the game based on the blueprint. Therefore, when using an informed prior, such as the blueprint, the resolving gadget game often outperforms the max-margin. 

%% file: theorems/maxmargin_worse.tex
\begin{restatable}[]{observation}{maxmarginworse}
\label{thm:maxmarginworse} 
    There exists a game $\FOSGame$, a subgame $\Subgame{}$ of $\FOSGame$, and a blueprint strategy $\BlueprintStrategy{\Player}$ of player $\Player$, in which some Nash equilibria of the corresponding resolving gadget game are less exploitable than any Nash equilibrium in the corresponding max-margin gadget game.
\end{restatable}

%% file: chapters/7_algorithms.tex
\section{Converging to Gadget Game Sequential Equilibria}
A key advantage of GGSE is that it can be computed using standard tabular algorithms with only minor modifications. We describe how to modify the sequence-form linear program and counterfactual regret minimization to converge to GGSE. Both of these approaches are based on the idea of assuming the opponent makes small mistakes in the auxiliary information sets $\AuxiliaryInfosets{}$
\subsection{Sequence-Form Linear Program}
Consider the SQF for finding Player 1's part of the Nash equilibrium \citep{koller1996}:
\begin{subequations}
\label{eq:sqf}
\begin{align}
    \max_{\SQFStrategy{1}, \SQFValue{2}} &\SQFInitRealization{2}^\Transposition \SQFValue{2} \label{eq:sqf_obj} \\
    \SQFRealizationConstraints{1} \SQFStrategy{1} &= \SQFInitRealization{1} \label{eq:sqf_primals} \\
    \SQFUtility{}^\Transposition  \SQFStrategy{1}  - \SQFRealizationConstraints{2}^\Transposition \SQFValue{2}  &\geq \boldsymbol{0} \label{eq:sqf_duals}  \\
    \SQFStrategy{1} &\geq \boldsymbol{0} \label{eq:sqf_probs}
\end{align}
\end{subequations}

Now consider a gadget game $\GadgetGame{}$ and a perturbation vector $\SQFPrior \in \mathbb{R}^{|\Infosets{2}|}$ that is zero everywhere except for sequences that end with action in the auxiliary information set $\Action{2} \in \Actions{2}(\AuxiliaryInfoset{2}) $, where it contains the prior $\StrategyPrior{2}(\AuxiliaryInfoset{2}, \Action{2})$. The following linear program approximates GGSE for a sufficiently small $\PriorEpsilon$:
\begin{subequations}
\label{eq:subgame_sqf}
\begin{align}
    \max_{\SQFStrategy{1}, \SQFValue{2}} & (\SQFInitRealization{2}^\Transposition  + \PriorEpsilon \SQFPrior^\Transposition ) \SQFValue{2} \label{eq:subgame_sqf_obj}  \\
    \SQFRealizationConstraints{1} \SQFStrategy{1} &= \SQFInitRealization{1} \label{eq:subgame_sqf_primals}  \\
    \SQFUtility{}^\Transposition  \SQFStrategy{1}  - \SQFRealizationConstraints{2}^\Transposition \SQFValue{2}  &\geq \boldsymbol{0} \label{eq:subgame_sqf_duals} \\
    \SQFStrategy{1} &\geq \boldsymbol{0}  \label{eq:subgame_sqf_probs}
\end{align}
\end{subequations}
 
This formulation resembles the LP used to find one-sided quasi-perfect equilibria \citep{farina2021osqpe}. It similarly does not require an arbitrarily small $\PriorEpsilon$, but it suffices to solve for a small $\PriorEpsilon$. If the solution remains optimal when $\PriorEpsilon$ is set to 0, then the solution is the desired equilibrium. Otherwise, the $\PriorEpsilon$ decreases, and the program is solved again. However, this LP presents two key distinctions over one-sided quasi-perfect equilibria. First, the trembles are assumed only in auxiliary information sets $\AuxiliaryInfosets{2}$ rather than every decision point. Second, within the same decision points, two actions may be played with different minimal probabilities, depending on the prior $\StrategyPrior{}$.

\subsection{Counterfactual Regret Minimization}
\label{sec:cfr_seq}
The majority of successful subgame-solving applications have used counterfactual regret minimization (CFR), which iteratively minimizes counterfactual regret at each decision node. 
To approximate GGSE with CFR, we adopt the approach of \citet{farina2017regret}, which modifies the Regret Matching within CFR framework by introducing the concept of basis matrix $\PerturbationBasis_{\Infoset{\Player}} \in \RealNumbers^{|\Actions{\Player}(\Infoset{\Player})| \times |\Actions{\Player}(\Infoset{\Player})|}$. Each entry in the matrix is non-negative, and each row and column sum to 1. This matrix serves as a transformation from the non-perturbed strategy to the perturbed one, and then transforms the regrets from the perturbed space to the non-perturbed one.
\begin{align}
    \Strategy{\Player}'(\Infoset{\Player}) = \PerturbationBasis_{\Infoset{\Player}} \Strategy{\Player}(\Infoset{\Player})\\
    \Regret{\Player}'(\Infoset{\Player}) = \PerturbationBasis_{\Infoset{\Player}}^\Transposition \Regret{\Player}(\Infoset{\Player})
\end{align} 

We use this basis in auxiliary information sets $\AuxiliaryInfoset{2}$:
\begin{align}
    \PerturbationBasis_{\AuxiliaryInfoset{2}} = (1 - \PriorEpsilon)\boldsymbol{I} + \PriorEpsilon \boldsymbol{P} 
\end{align}
$\boldsymbol{I}$ is an identity and $\boldsymbol{P}$ is a symmetric matrix where a column corresponding to action $\Action{2} \in \Actions{2}(\AuxiliaryInfoset{2})$ contains $\StrategyPrior{2}(\AuxiliaryInfoset{2}, \Action{2})$.

This change does not significantly increase per-iteration time, but converging to an equilibrium may require more iterations based on the $\PriorEpsilon$ and $\boldsymbol{P}$, which we show in \cref{sec:experiments}.

When using either the CFR or SQF, the resulting equilibrium is a GGSE, but it need not be consistent with the prior $\StrategyPrior{}$. The GGSE is consistent with the prior only if the equilibrium itself does not play into any part of the subgame. The solutions of the max-margin gadget game wuill always be inconsistent, but those of the resolving gadget game can be consistent (exactly when all the counterfactual values are lower than the optimal ones).

%% file: chapters/8_experiments.tex
\section{Experiments}
\label{sec:experiments}

\begin{figure}[]
    \centering
        \begin{subfigure}[b]{0.23\textwidth}
            \centering
            \includegraphics[width=\textwidth]{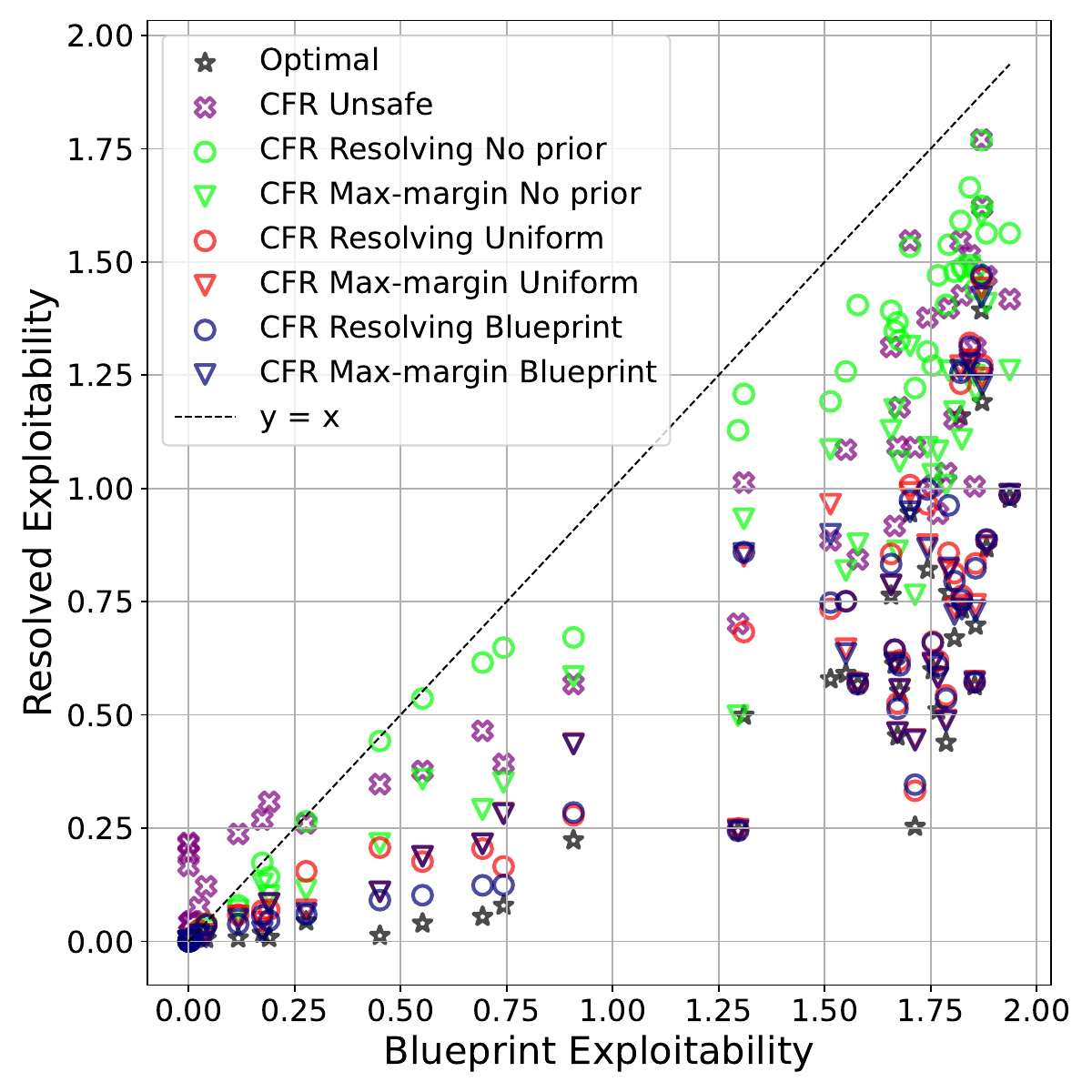} 
            \caption{Goofspiel 5, depth 1}
            \label{fig:expl_cfr_goof5_d1}
        \end{subfigure}%
        \hfill
        \begin{subfigure}[b]{0.23\textwidth}
            \centering
            \includegraphics[width=\textwidth]{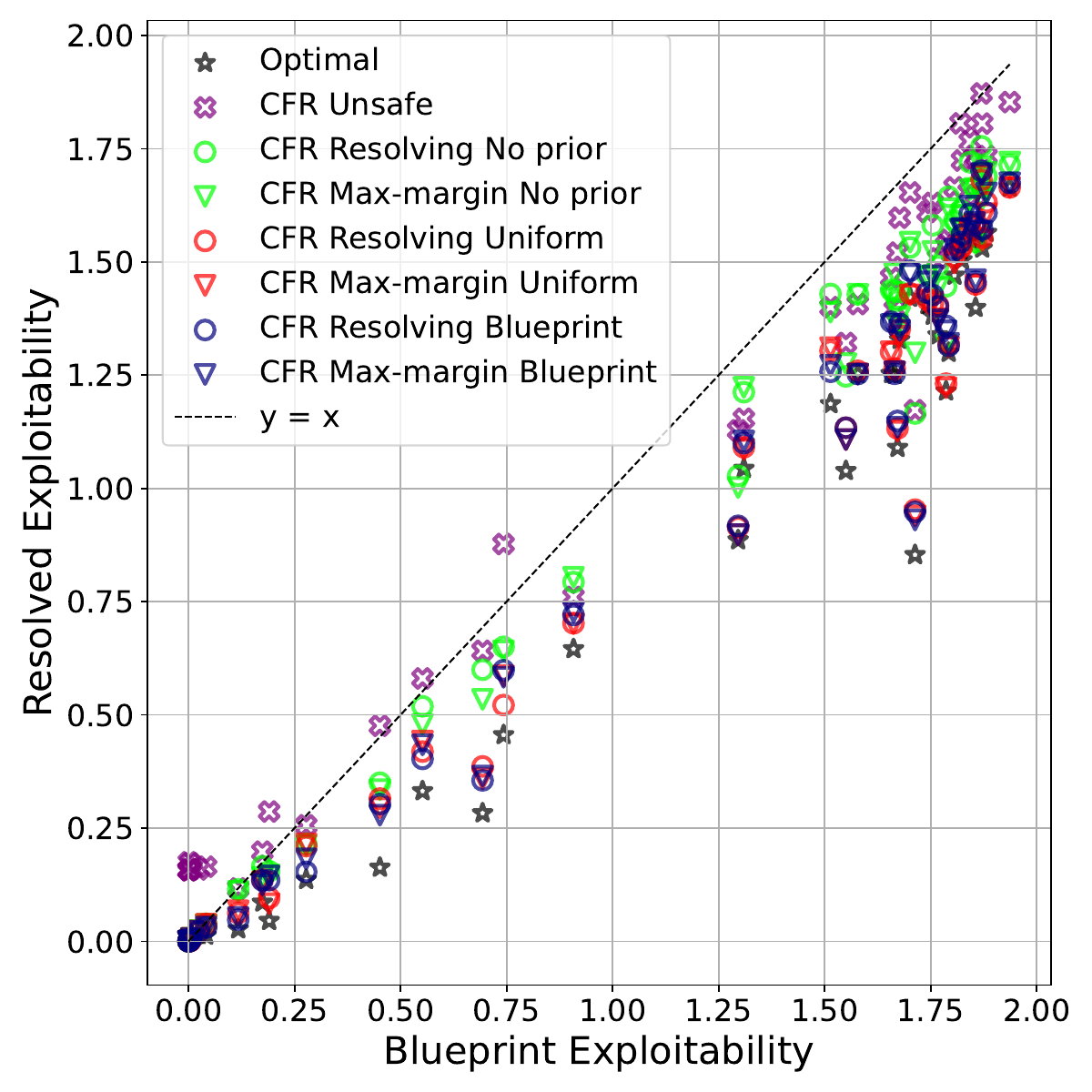} 
            \caption{Goofspiel 5, depth 2}
            \label{fig:expl_cfr_goof5_d2}
        \end{subfigure}
        \begin{subfigure}[b]{0.23\textwidth}
            \centering
            \includegraphics[width=\textwidth]{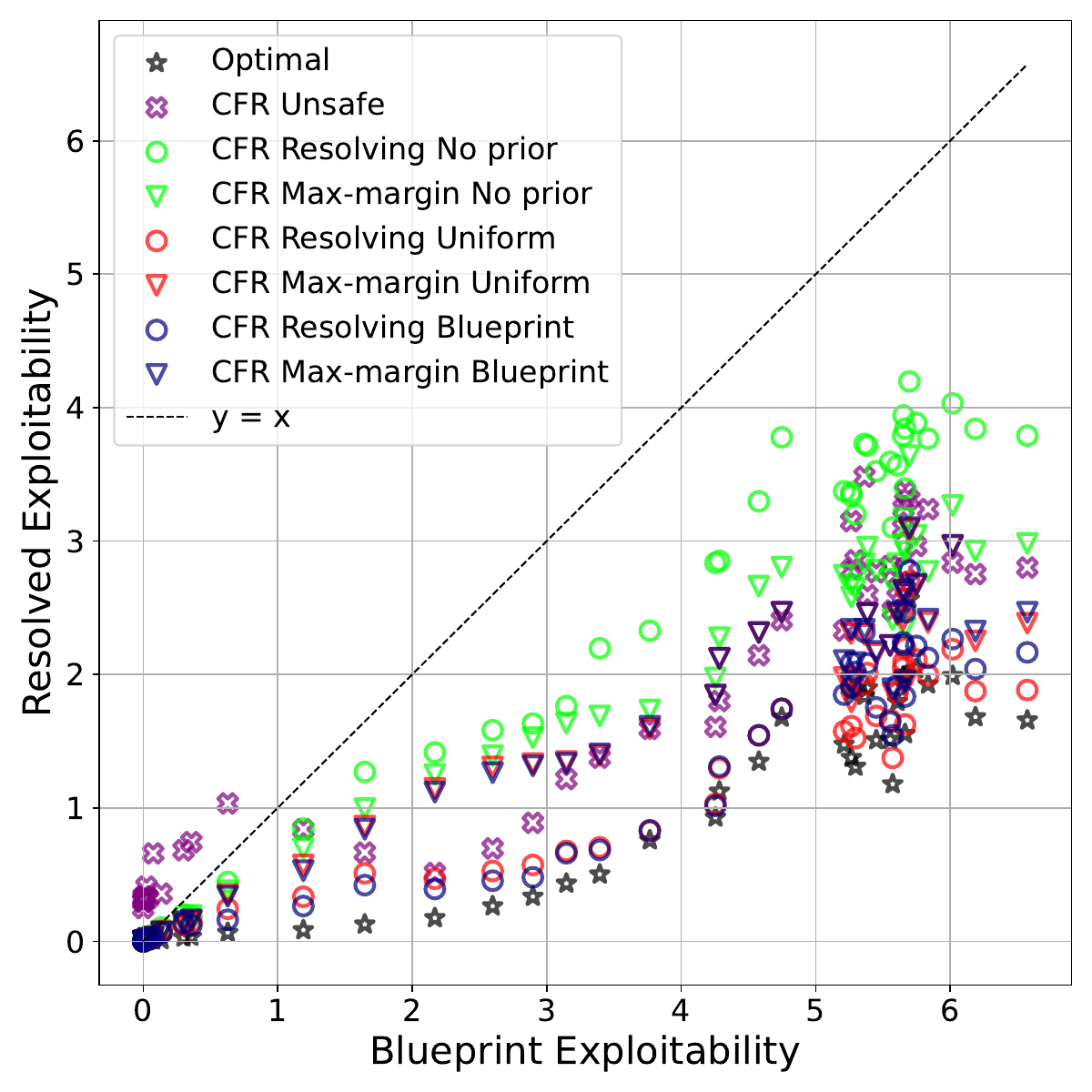} 
            \caption{Leduc hold'em after flop}
            \label{fig:expl_cfr_leduc}
        \end{subfigure}%
        \hfill
        \begin{subfigure}[b]{0.23\textwidth}
            \centering
            \includegraphics[width=\textwidth]{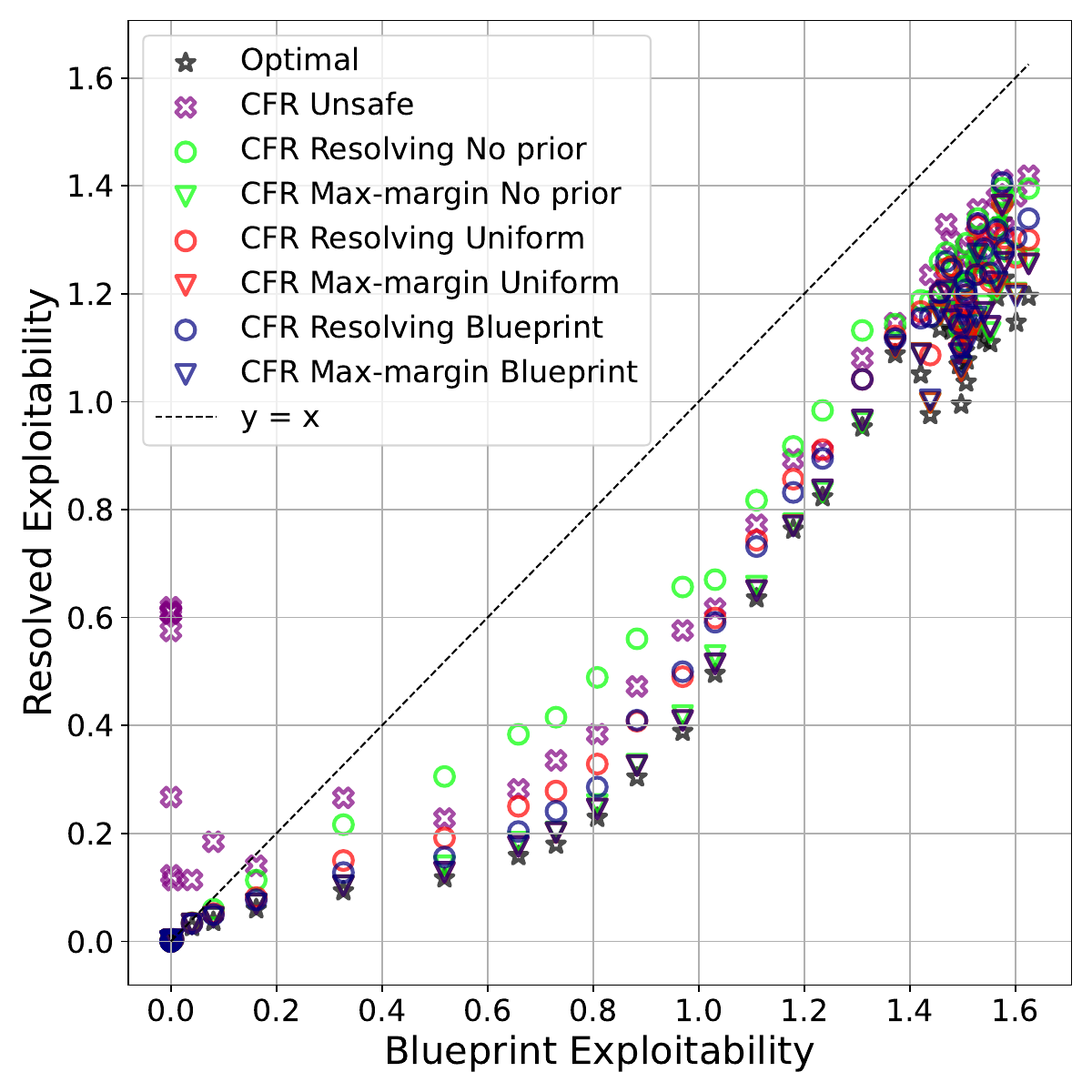} 
            \caption{Liar's dice 1,4, depth 1}
            \label{fig:expl_cfr_liars_d1}
        \end{subfigure}
    \caption{Exploitability of the resolved strategy from each subgame based on blueprint's exploitability when using CFR as solver.}
    \label{fig:resolved_exploitability_cfr}
\end{figure}

We evaluated the proposed algorithms for finding GGSE across three standard benchmark games: imperfect-information Goofspiel 5, Leduc hold'em, and Liar's dice 1,4. In each of those games, we generated 25 blueprints from different checkpoints of CFR. These serve as realistic approximations of blueprint strategies typically produced by abstraction-based methods or Monte Carlo CFR. We also generated 20 blueprints, in which the behavioral strategy in each information set was drawn independently from the Dirichlet distribution $\text{Dir}(\boldsymbol{\alpha})$, with $\alpha_i=1$.

We solve subgames at specific decision points. In Goofspiel and Liar's dice, we solve every subgame following a fixed number of moves, and in Leduc hold'em, we solve all subgames after the public card is revealed. For each game and blueprint, we computed different types of subgame equilibria using the CFR modification discussed in \cref{sec:cfr_seq}, with CFR+ updates \citep{cfrplus2014}. We compare three different prior settings: no prior, uniform prior, or blueprint prior. For comparison, we also include the optimal continuation strategy and unsafe solving, which assumes both players play a fixed strategy before the subgame.

The relationship between the exploitability of the resolved strategy and the exploitability of the blueprint is shown in \cref{fig:resolved_exploitability_cfr}. In \cref{app:experiments}, we provide statistical analysis for our results and the same results when using SQF as a solver. We also show the convergence curves in \cref{fig:convergence} for one of the blueprints to demonstrate the computational overhead in convergence when using a prior. We highlight the following results from this experiment.
\begin{itemize}
    \item Incorporating the prior does not increase the exploitability and often decreases it. The gains are more significant when using the resolving gadget game.
    \item In many cases, the max-margin solutions outperform the resolving ones when no prior is used, in line with prior work. However, when using the prior, in Leduc hold'em, the solutions of the resolving gadget game outperform the max-margin ones, confirming that max-margin does not have to be superior to resolving.
    \item Although unsafe solving performs well when the blueprint is poor, its performance degrades as the blueprint improves. In contrast, when using the resolving gadget game with a blueprint prior, the solution achieves performance comparable to unsafe solving while being safe, effectively combining the strengths of both approaches.
    \item Converging to GGSE with CFR requires more iterations. Even with fewer iterations, the performance does not deteriorate compared to not using a prior at all, suggesting that using the prior does not come with any downside.
\end{itemize}

\begin{figure}[]
    \centering
        \begin{subfigure}[b]{0.24\textwidth}
            \centering
            \includegraphics[width=\textwidth]{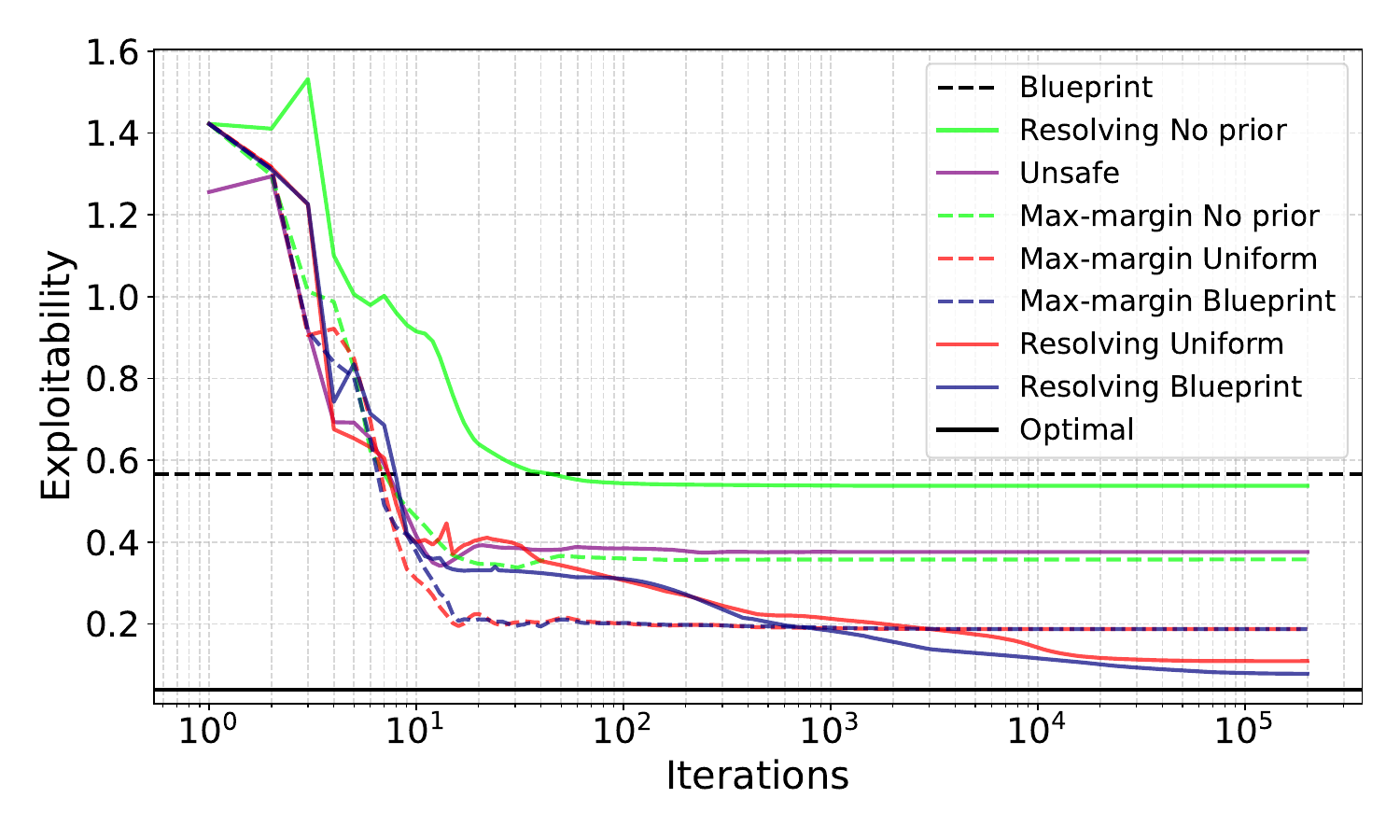} 
            \caption{Goofspiel 5, depth 1}
            \label{fig:conv_goof}
        \end{subfigure}%
        \hfill
        \begin{subfigure}[b]{0.24\textwidth}
            \centering
            \includegraphics[width=\textwidth]{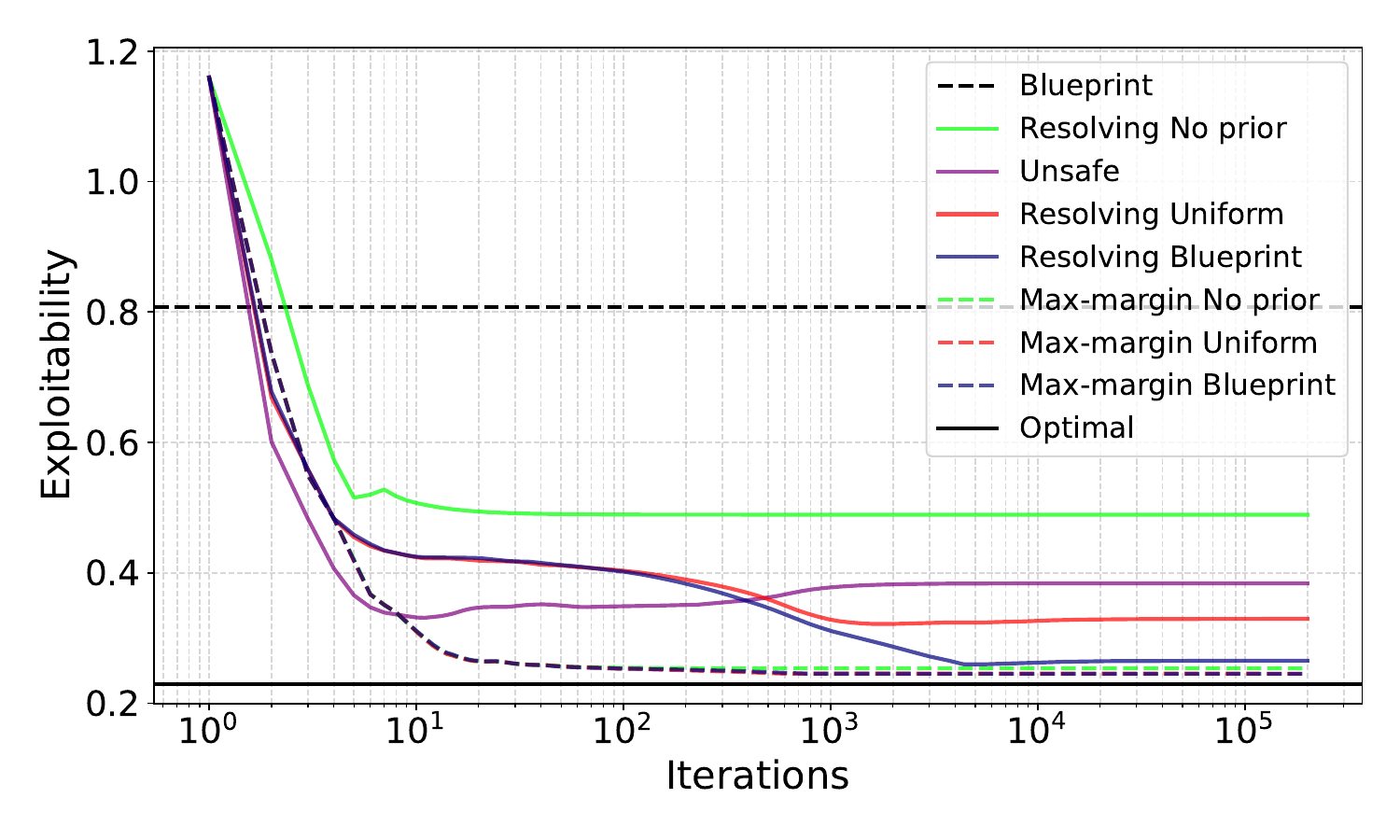} 
            \caption{Liar's dice 1,4, depth 1}
            \label{fig:conv_liars}
        \end{subfigure}
    \caption{CFR Convergence curves of subgame solving techniques.}
    \label{fig:convergence}
\end{figure}

\subsection{Nested Depth-Limited Subgame Solving}
\label{sec:nested}
In large games, it is intractable to solve the whole game from the current decision point all the way to terminal states. The common approach to address this is to use nested depth-limited subgame solving, also known as midgame solving~\citep{brown2018depth,pluribus2019,sepot2024}. This approach trains a value function that is used after a fixed number of moves. In our experiments, we used a value function based on multi-valued states, which, after the depth limit, abstracts the entire rest of the game into a single decision by both players, in which they choose a strategy for the remainder of the game~\citep{pluribus2019,valuefunctions2023}.

We trained 10 different value functions in imperfect-information Goofspiel 7 and Liar's dice 2,4 using the approach of \citet{sepot2024}. Then we matched different types of subgame solving at each encountered decision point. Head-to-head performance does not directly correspond to exploitability, because neither technique aims to exploit the opponent; rather, each approximates the Nash equilibrium of the entire game. \cref{fig:heads} shows the win-rate of each subgame solving technique against the worst matchup out of all the other techniques, which serves as a very limited approximation of exploitability. We provide the full head-to-head win-rate for each pair of techniques in \cref{app:experiments}.

The results confirm that using the prior is beneficial even in head-to-head play. In Goofspiel, the unsafe solving performs well, while the solutions of the max-margin gadget game perform poorly, regardless of the prior. The situation in Liar's dice is the opposite. In both games, the solutions from the resolving gadget game with a blueprint prior consistently perform well across opponents.
\begin{figure}
    \centering
    \includegraphics[width=0.97\linewidth]{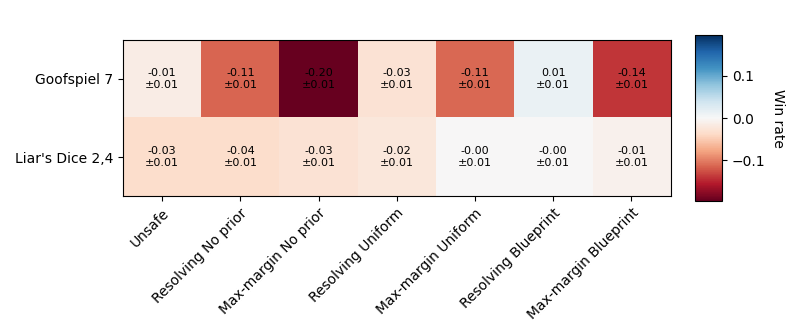}
    \caption{Head-to-head win-rate of each subgame solving technique against its 'worst-matchup' opponent out of the other techniques. }
    \label{fig:heads}
\end{figure}

%% file: chapters/9_conclusion.tex
\section{Conclusions and Future Research}
We studied gadget games, a subgame solving technique for imperfect-information games. While gadget games provide essential safety guarantees, they often contain infinitely many Nash equilibria. We showed that those equilibria may seem equivalent in the context of the gadget game, but when used in the original game, they lead to vastly different exploitability. 

We proposed a new equilibrium refinement, specifically tailored to gadget games, \textit{gadget game sequential equilibrium}, that improves upon vanilla Nash equilibria. We provided modifications to both the sequence-form linear program and counterfactual regret minimization to compute those equilibria refinements with negligible additional computational overhead. 

We provided a novel explanation for the empirical superiority of the resolving gadget game over max-margin. We showed that the max-margin gadget game introduces an inherent bias into the solution, potentially excluding good equilibria that remain solutions of the resolving gadget game. We empirically verified that this occurs regularly in some games.

We compared different subgame-solving methods across multiple benchmarks and verified that converging to gadget game sequential equilibria consistently performs as well or better than both unsafe solving and the usual unrefined Nash equilibria to which CFR and the sequence form LP converge in gadget games. 
Thanks to improved theoretical understanding and empirical performance, this work establishes gadget game sequential equilibrium as the preferred solution concept for subgame solving in imperfect-information games.  

Even though we established that there are games where neither the resolving nor the max-margin gadget game have the optimal continuation strategy as a solution, there is potential future work to either develop a new gadget game or to modify the existing ones so that they can reconstruct an optimal continuation strategy without the linear growth in size of the game as ``full gadget''.
Another future research direction is to verify whether gadget game sequential equilibria present similar performance gains even when using the reach max-margin gadget game and/or when conducting knowledge-limited subgame solving. Also, there is potential to apply gadget game sequential equilibria to subgame solving in large recreational games like Texas hold'em and dark chess as well as to real-world game applications.

%% file: chapters/acknowledgments.tex
\section*{Acknowledgments}
This work is supported by National Science Foundation grant RI-2312342 , the Vannevar Bush Faculty Fellowship ONR N00014-23-1-2876, the Czech Science Foundation GA25-18353S, the Grant Agency of the Czech Technical University in Prague (SGS23/184/OHK3/3T/13). The access to the computational infrastructure of the OP VVV funded project CZ.02.1.01/0.0/0.0/16\_019/0000765 ``Research Center for Informatics'' is also gratefully acknowledged.  Any opinions, findings, and conclusions or recommendations expressed in this material are those of the author(s) and do not necessarily reflect the views of the funding agencies.

%% file: appendix.tex
\input{chapters/proofs}
\input{chapters/add_experiments}
\input{chapters/game_rules}
\input{chapters/sequence_form}

\input{chapters/llm}

%% file: chapters/proofs.tex
\section{Theorems and Observations}
\label{app:proofs}
\gadgetvalue*
\input{theorems/gadget_value_proof}

\gadgetequilibria*

\input{theorems/gadget_equilibria_proof}

\ggse*
\input{theorems/ggse_superiority_proof}

\nonexistence*
\input{theorems/non_existence_proof}

\maxmarginworse*
\input{theorems/maxmargin_worse_proof}

%% file: theorems/gadget_value_proof.tex
{\allowdisplaybreaks
\begin{proof}



We will give the proof from the perspective of Player 1. The proof for Player 2 is identical. We will use $\WeightedReward{1}{\Strategy{}}(\History, \Action{}) := \Reach{\Strategy{}}(\History) \prod_{\Player \in \Players}\Strategy{\Player}(\Infoset{\Player}(\History), \Action{\Player}) \Rewards{1}(\History, \Action{})$, where $\Action{\Player} $ is the Player's $\Player$ part of action $\Action{}$.

The proof shows that when replacing the strategy from a subgame, the exploitability is influenced only by this subgame. Then it shows that if counterfactual best-response values increase in the root information sets, the exploitability decreases. 

Consider a strategy $\Strategy{2}'$, which plays $\BRStrategy{2}$ everywhere except the parts within subgame $\Subgame{}$, where it plays $\BRBlueprintStrategy{2}$. 




\begin{align}
    &\Utility{1}^{(\Strategy{1}, \BRStrategy{-1})}(\InitHistory) = \sum_{\History \HistoryExtend \History' \Action{}'}  \WeightedReward{1}{(\Strategy{1}, \BRStrategy{-1})} (\History', \Action{}') \\
    & = \sum_{\substack{\History \HistoryExtend \History' \Action{}' \\ \History \in \Subgame{}}} \WeightedReward{1}{(\Strategy{1}, \BRStrategy{-1})} (\History', \Action{}') +\sum_{\substack{\History \HistoryExtend \History' \Action{}' \\ \History \notin \Subgame{}}} \WeightedReward{1}{(\Strategy{1}, \BRStrategy{-1})} (\History', \Action{}') \\
    & = \sum_{\substack{\History \HistoryExtend \History' \Action{}' \\ \History \in \Subgame{}}}\WeightedReward{1}{(\Strategy{1}, \BRStrategy{-1})}  (\History', \Action{}') + \sum_{\substack{\History \HistoryExtend \History' \Action{}' \\ \History \notin \Subgame{}}} \WeightedReward{1}{(\BlueprintStrategy{1},  \BRStrategy{-1})} (\History', \Action{}') \\
    & = \sum_{\substack{\History \HistoryExtend \History' \Action{}' \\ \History \in \Subgame{}}}\WeightedReward{1}{(\Strategy{1}, \BRStrategy{-1})}  (\History', \Action{}') - \sum_{\substack{\History \HistoryExtend \History' \Action{}' \\ \History \in \Subgame{}}} \WeightedReward{1}{(\BlueprintStrategy{1},  \Strategy{-1}')}(\History', \Action{}')\\
    &+ \sum_{\substack{\History \HistoryExtend \History' \Action{}' \\ \History \in \Subgame{}}} \WeightedReward{1}{(\BlueprintStrategy{1},  \Strategy{-1}')}(\History', \Action{}') + \sum_{\substack{\History \HistoryExtend \History' \Action{}' \\ \History \notin \Subgame{}}} \WeightedReward{1}{(\BlueprintStrategy{1},  \BRStrategy{-1})}(\History', \Action{}') \\
    & \geq \sum_{\substack{\History \HistoryExtend \History' \Action{}' \\ \History \in \Subgame{}}}\WeightedReward{1}{(\Strategy{1}, \BRStrategy{-1})}  (\History', \Action{}') - \sum_{\substack{\History \HistoryExtend \History' \Action{}' \\ \History \in \Subgame{}}} \WeightedReward{1}{(\BlueprintStrategy{1},  \Strategy{-1}')}(\History', \Action{}')\\
    &+  \sum_{\substack{\History \HistoryExtend \History' \Action{}' \\ \History \in \Subgame{}}} \WeightedReward{1}{(\BlueprintStrategy{1},  \BRBlueprintStrategy{-1})}(\History', \Action{}')\sum_{\substack{\History \HistoryExtend \History' \Action{}' \\ \History \notin \Subgame{}}} \WeightedReward{1}{(\BlueprintStrategy{1},  \BRBlueprintStrategy{-1})}(\History', \Action{}') \\
    & = \sum_{\substack{\History \HistoryExtend \History' \Action{}' \\ \History \in \Subgame{}}}\WeightedReward{1}{(\Strategy{1}, \BRStrategy{-1})}  (\History', \Action{}') - \sum_{\substack{\History \HistoryExtend \History' \Action{}' \\ \History \in \Subgame{}}} \WeightedReward{1}{(\BlueprintStrategy{1},  \Strategy{-1}')}(\History', \Action{}')\\
    & + \Utility{1}^{(\BlueprintStrategy{1}, \BRBlueprintStrategy{-1})}(\InitHistory) \\ 
    & = \sum_{\Infoset{2} \in \Infosets{2}(\PublicState)} \sum_{\History \in \Histories(\Infoset{2})} \PlayerReach{1}{\BlueprintStrategy{1}}(\History) \PlayerReach{2}{\BRStrategy{2}}(\History) \PlayerReach{\ChancePlayer}{\Strategy{\ChancePlayer}}(\History) \\
    &\bigg(\Utility{1}^{(\Strategy{1}, \BRStrategy{-1})}(\History) - \Utility{1}^{(\BlueprintStrategy{1}, \BRBlueprintStrategy{-1})}(\History) \bigg)  + \Utility{1}^{(\BlueprintStrategy{1}, \BRBlueprintStrategy{-1})}(\InitHistory) \\
    & = \sum_{\Infoset{2} \in \Infosets{2}(\PublicState)} \PlayerReach{2}{\BRStrategy{2}}(\Infoset{2}) \sum_{\History \in \Histories(\Infoset{2})}\PlayerReach{1}{\BlueprintStrategy{1}}(\History) \PlayerReach{\ChancePlayer}{\Strategy{\ChancePlayer}}(\History)\\
    &\bigg(\Utility{1}^{(\Strategy{1}, \BRStrategy{-1})}(\History) - \Utility{1}^{(\BlueprintStrategy{1}, \BRBlueprintStrategy{-1})}(\History) \bigg)  + \Utility{1}^{(\BlueprintStrategy{1}, \BRBlueprintStrategy{-1})}(\InitHistory) \\
    & = \sum_{\Infoset{2} \in \Infosets{2}(\PublicState)} \PlayerReach{2}{\BRStrategy{2}}(\Infoset{2}) \bigg(\CounterfactualValue{1}{(\Strategy{1}, \BRStrategy{-1})}(\Infoset{2}) - \CounterfactualValue{1}{(\BlueprintStrategy{1}, \BRBlueprintStrategy{-1})}(\Infoset{2}) \bigg)\\
    &+ \Utility{1}^{(\BlueprintStrategy{1}, \BRBlueprintStrategy{-1})}(\InitHistory)
\end{align}

 As a result, if $\PlayerReach{2}{\Strategy{2}^{'BR}}(\Infoset{2}) > 0$, the Player 2 plays into the information set $\Infoset{-\Player}$ with the best response. This results in an improvement 
 \begin{align}
     &\Utility{1}^{(\Strategy{1}, \BRStrategy{-1})}(\InitHistory) - \Utility{1}^{(\BlueprintStrategy{1}, \BRBlueprintStrategy{-1})}(\InitHistory) \\
     \geq &\sum_{\Infoset{2} \in \Infosets{2}(\PublicState)} \PlayerReach{2}{\BRStrategy{2}}(\Infoset{2}) \bigg(\CounterfactualValue{1}{(\Strategy{1}, \BRStrategy{-1})}(\Infoset{2}) - \CounterfactualValue{1}{(\BlueprintStrategy{1}, \BRBlueprintStrategy{-1})}(\Infoset{2}) \bigg)
 \end{align} 
\end{proof}
}

%% file: theorems/gadget_equilibria_proof.tex
\begin{proof}

\begin{figure}
    \centering
    \input{tikz/gadget_equilibrium_example}
    \caption{An example of a game, where different equilibria in gadget games result in different utility}
    \label{fig:gadet_equilibruim_example}
\end{figure}

\begin{figure}
    \centering
    \input{tikz/gadget_equilibrium_resolving}
    \caption{Resolving gadget game of subgame from \cref{fig:gadet_equilibruim_example}}
    \label{fig:gadet_equilibruim_resolving}
\end{figure}

\begin{figure}
    \centering
    \input{tikz/gadget_equilibrium_maxmargin}
    \caption{Max-margin gadget game of subgame from \cref{fig:gadet_equilibruim_example}}
    \label{fig:gadet_equilibruim_maxmargin}
\end{figure}

    Consider a game that starts with a chance node that uniformly selects one of the 3 states. The Player 1 cannot distinguish states 2 and 3, whereas Player 2 cannot distinguish states 1 and 2. States 2 and 3 are then followed by the Matching Pennies game. State 1 is followed by a modified Matching Pennies, where each outcome gives Player 1 one more point, except when Player 1 plays Tails, and Player 2 plays Heads. We show an extensive-form representation of this game in \cref{fig:gadet_equilibruim_example}.
    
    This games contains a single Nash equilibrium $\Nash{1}(\Infoset{1}^1, H) = \frac{2}{3}$, $\Nash{1}(\Infoset{1}^2, H) = \frac{1}{2}$, $\Nash{2}(\Infoset{2}^1, H) = \frac{1}{3}$, $\Nash{2}(\Infoset{2}^2, H) = \frac{2}{3}$ with value $\Utility{\Player}^{\Nash{}}(\InitHistory) = \frac{7}{9}$.

    Assume a blueprint $\BlueprintStrategy{1}$, which is the same as $\Nash{1}$ in $\Infoset{1}^1$, but differs in $\Infoset{1}^{2}$ in following way $\BlueprintStrategy{1}(\Infoset{1}^1, H) = \frac{2}{3}$ and $\BlueprintStrategy{1}(\Infoset{1}^2, H) = 0$. Assume a $\BRBlueprintStrategy{2} \in \BRSet{2}(\BlueprintStrategy{1})$, then the expected utility and exploitability is $\Utility{1}^{(\BlueprintStrategy{1}, \BRBlueprintStrategy{2})}(\InitHistory) = \frac{1}{3}$,  $\Exploitability(\BlueprintStrategy{1}) = \frac{7}{9} - \frac{1}{3} = \frac{4}{9}$

    Now assume Player 1 reached information set $\Infoset{1}^1$ and wants to use common-knowledge subgame solving. Since Player 2 cannot distinguish between states 1 and 2, the common-knowledge subgame must contain both. Similarly, Player 1 cannot distinguish between states 2 and 3. So the subgame contains all 3 states. The counterfactual best response values in the opponent's information sets are
    \begin{align}
        \CounterfactualValue{1}{\BlueprintStrategy{-2}}(\Infoset{2}^1) &= 0.25 \\
        \CounterfactualValue{1}{\BlueprintStrategy{-2}}(\Infoset{2}^2) &= 0.5        
    \end{align}

    The extensive-form representations of resolving and max-margin gadget games are represented in \cref{fig:gadet_equilibruim_resolving,fig:gadet_equilibruim_maxmargin}. Starting with the resolving gadget game. In every Nash equilibrium, Player 2 plays $\Strategy{2}(\Infoset{2}^{G, 1}, T) = 1.0$, as player 1 can always guarantee to get at least utility 1 by playing $H$ in the left-most state. Therefore, the strategy in $\Infoset{1}^{2}$ depends only on the right-most node, as it may be reached with non-zero probability. Since this is just a matching pennies game, the optimal strategy of Player 1 has to always be $\Strategy{1}(\Infoset{1}^2, H) = \frac{1}{2}$. Player 1 can play any strategy in $\Infoset{1}^1$, and it will still be part of a Nash equilibrium. This means that both $\BlueprintStrategy{1}$ and $\Nash{1}$ are Nash equilibria of the resolving gadget game. However, $\Exploitability(\BlueprintStrategy{1}) = \frac{4}{9}$, but $\Exploitability(\Nash{1}) = 0$.

    In the max-margin gadget game, Player 2 will always play to the right subtree, as it can guarantee a reward $0$ there. Similarly to resolving gadget, any Nash equilibrium in the max-margin gadget game needs to satisfy $\Strategy{1}(\Infoset{1}^2, H) = \frac{1}{2}$, but the strategy in $\Infoset{1}^1$ can be arbitrary. Again, both $\BlueprintStrategy{1}$ and $\Nash{1}$ are Nash equilibria of the max-margin gadget game, but they result in different exploitability in the original game.
\end{proof}

%% file: tikz/gadget_equilibrium_example.tex
\begin{tikzpicture}[
    level 1/.style={sibling distance=3cm, level distance=1cm}, 
    level 2/.style={sibling distance=1.4cm, level distance=1cm},
    level 3/.style={sibling distance=1cm, level distance=1cm},
    p2node/.style={ 
      regular polygon,
      regular polygon sides=3,
      rotate=180, 
      draw,
      fill=blue!40,
      inner sep=1pt,
      minimum size=15pt
    },
    p1node/.style={
      regular polygon,
      regular polygon sides=3, 
      draw,
      fill=red!40,
      inner sep=1pt,
      minimum size=15pt
    },
    chancenode/.style={circle, draw, fill=gray!30, minimum size=5mm, inner sep=1pt},
    terminal/.style={inner sep=2pt},
    payoff/.style={
      font=\normalsize
    },
    edge_label/.style={midway, fill=white, inner sep=1pt}
]

\node[chancenode] (S) {} 
    child {
        node[p1node] (P1-L) {}
        child { 
            node[p2node] (P2-L1) {}
            child {
                node[payoff] (T1) {2}
                edge from parent
                node[edge_label,above left] {H}
            }
            child {
                node[payoff] (T2) {1}
                edge from parent
                node[edge_label,above right] {T}
            } 
            edge from parent
            node[edge_label,above left] {H}
        }
        child { 
            node[p2node] (P2-L2) {}
            child {
                node[payoff] (T3) {0}
                edge from parent
                node[edge_label,above left] {H}
            }
            child {
                node[payoff] (T4) {2}
                edge from parent
                node[edge_label,above right] {T}
            } 
            edge from parent
            node[edge_label,above right] {T}
        }
        edge from parent
        node[edge_label, above left] {$\frac{1}{3}$}
    }
    child {
        node[p1node] (P1-M) {}
        child { 
            node[p2node] (P2-M1) {}
            child {
                node[payoff] (T5) {1}
                edge from parent
                node[edge_label,above left] {H}
            }
            child {
                node[payoff] (T6) {0}
                edge from parent
                node[edge_label,above right] {T}
            } 
            edge from parent
            node[edge_label, above left] {H}
        }
        child { 
            node[p2node] (P2-M2) {}
            child {
                node[payoff] (T7) {0}
                edge from parent
                node[edge_label,above left] {H}
            }
            child {
                node[payoff] (T8) {1}
                edge from parent
                node[edge_label,above right] {T}
            } 
            edge from parent
            node[edge_label, above right] {T}
        }
        edge from parent
        node[edge_label, left] {$\frac{1}{3}$}
    }
    child {
        node[p1node] (P1-R) {}
        child { 
            node[p2node] (P2-R1) {}
            child {
                node[payoff] (T9) {1}
                edge from parent
                node[edge_label,above left] {H}
            }
            child {
                node[payoff] (T10) {0}
                edge from parent
                node[edge_label,above right] {T}
            } 
            edge from parent
            node[edge_label,above left] {H}
        }
        child { 
            node[p2node] (P2-R2) {}
            child {
                node[payoff] (T11) {0}
                edge from parent
                node[edge_label,above left] {H}
            }
            child {
                node[payoff] (T12) {1}
                edge from parent
                node[edge_label,above right] {T}
            } 
            edge from parent
            node[edge_label,above right] {T}
        }
        edge from parent
        node[edge_label, above right] {$\frac{1}{3}$}
    }; 
  \draw [red, dashed, thick] (P1-M) -- (P1-R)
        node [midway, above left, yshift=1mm, fill=white, inner sep=1pt, font=\small] {$\Infoset{1}^2$};
        
  \draw [blue, dashed, thick] (P2-L1) -- (P2-L2)
        node [midway, above, yshift=1mm, fill=white, inner sep=1pt, font=\small] {$\Infoset{2}^1$};
  \draw [blue, dashed, thick] (P2-L2) -- (P2-M1)
        node [midway, above, yshift=1mm, fill=white, inner sep=1pt, font=\small] {$\Infoset{2}^1$};

  \draw [blue, dashed, thick] (P2-M1) -- (P2-M2)
        node [midway, above, yshift=1mm, fill=white, inner sep=1pt, font=\small] {$\Infoset{2}^1$};
        
  \draw [blue, dashed, thick] (P2-R1) -- (P2-R2)
        node [midway, above, yshift=1mm, fill=white, inner sep=1pt, font=\small] {$\Infoset{2}^2$};
        
\end{tikzpicture}

%% file: tikz/gadget_equilibrium_resolving.tex
\begin{tikzpicture}[
    level 1/.style={sibling distance=3cm, level distance=1cm},
    level 2/.style={sibling distance=1cm, level distance=1cm}, 
    level 3/.style={sibling distance=1.4cm, level distance=1cm},
    level 4/.style={sibling distance=1cm, level distance=1cm},
    p2node/.style={ 
      regular polygon,
      regular polygon sides=3,
      rotate=180, 
      draw,
      fill=blue!40,
      inner sep=1pt,
      minimum size=15pt
    },
    p1node/.style={
      regular polygon,
      regular polygon sides=3, 
      draw,
      fill=red!40,
      inner sep=1pt,
      minimum size=15pt
    },
    chancenode/.style={circle, draw, fill=gray!30, minimum size=5mm, inner sep=1pt},
    terminal/.style={inner sep=2pt},
    payoff/.style={
      font=\normalsize
    },
    edge_label/.style={midway, fill=white, inner sep=1pt}
]

\node[chancenode] (S) {} 
    child {
        node[p2node] (P2-R-L) {}
        child[grow=-150]  { 
            node[payoff] (TER1) {0.25}
            edge from parent
            node[edge_label, above left] {T}
        }
        child[grow=-90]  {
            node[p1node] (P1-L) {}
            child { 
                node[p2node] (P2-L1) {}
                child {
                    node[payoff] (T1) {2}
                    edge from parent
                    node[edge_label,above left] {H}
                }
                child {
                    node[payoff] (T2) {1}
                    edge from parent
                    node[edge_label,above right] {T}
                } 
                edge from parent
                node[edge_label,above left] {H}
            }
            child { 
                node[p2node] (P2-L2) {}
                child {
                    node[payoff] (T3) {0}
                    edge from parent
                    node[edge_label,above left] {H}
                }
                child {
                    node[payoff] (T4) {2}
                    edge from parent
                    node[edge_label,above right] {T}
                } 
                edge from parent
                node[edge_label,above right] {T}
            }
            edge from parent
            node[edge_label, left] {C}
        }
        edge from parent
        node[edge_label, above left] {$\frac{1}{3}$}
    }
    child {
        node[p2node] (P2-R-M) {}
        child[grow=-150]  { 
            node[payoff] (TER2) {0.25}
            edge from parent
            node[edge_label, above left] {T}
        }
        child[grow=-90] {
            node[p1node] (P1-M) {}
            child { 
                node[p2node] (P2-M1) {}
                child {
                    node[payoff] (T5) {1}
                    edge from parent
                    node[edge_label,above left] {H}
                }
                child {
                    node[payoff] (T6) {0}
                    edge from parent
                    node[edge_label,above right] {T}
                } 
                edge from parent
                node[edge_label, above left] {H}
            }
            child { 
                node[p2node] (P2-M2) {}
                child {
                    node[payoff] (T7) {0}
                    edge from parent
                    node[edge_label,above left] {H}
                }
                child {
                    node[payoff] (T8) {1}
                    edge from parent
                    node[edge_label,above right] {T}
                } 
                edge from parent
                node[edge_label, above right] {T}
            }
            edge from parent
            node[edge_label, left] {C}
        }
        edge from parent
        node[edge_label, left] {$\frac{1}{3}$}
    }
    child {
        node[p2node] (P2-R-R) {}
        child[grow=-150]  { 
            node[payoff] (TER3) {0.5}
            edge from parent
            node[edge_label, above left,xshift=-1mm] {T}
        }
        child[grow=-90] {
            node[p1node] (P1-R) {}
            child { 
                node[p2node] (P2-R1) {}
                child {
                    node[payoff] (T9) {1}
                    edge from parent
                    node[edge_label,above left] {H}
                }
                child {
                    node[payoff] (T10) {0}
                    edge from parent
                    node[edge_label,above right] {T}
                } 
                edge from parent
                node[edge_label,above left] {H}
            }
            child { 
                node[p2node] (P2-R2) {}
                child {
                    node[payoff] (T11) {0}
                    edge from parent
                    node[edge_label,above left] {H}
                }
                child {
                    node[payoff] (T12) {1}
                    edge from parent
                    node[edge_label,above right] {T}
                } 
                edge from parent
                node[edge_label,above right] {T}
            }
            edge from parent
            node[edge_label, left] {C}
        }
        edge from parent
        node[edge_label, above right] {$\frac{1}{3}$}
    }; 
  \draw [red, dashed, thick] (P1-M) -- (P1-R)
        node [midway, above left, yshift=1mm, fill=white, inner sep=1pt, font=\small] {$\Infoset{1}^2$};
        
  \draw [blue, dashed, thick] (P2-L1) -- (P2-L2)
        node [midway, above, yshift=1mm, fill=white, inner sep=1pt, font=\small] {$\Infoset{2}^1$};
  \draw [blue, dashed, thick] (P2-L2) -- (P2-M1)
        node [midway, above, yshift=1mm, fill=white, inner sep=1pt, font=\small] {$\Infoset{2}^1$};

  \draw [blue, dashed, thick] (P2-M1) -- (P2-M2)
        node [midway, above, yshift=1mm, fill=white, inner sep=1pt, font=\small] {$\Infoset{2}^1$};
        
  \draw [blue, dashed, thick] (P2-R1) -- (P2-R2)
        node [midway, above, yshift=1mm, fill=white, inner sep=1pt, font=\small] {$\Infoset{2}^2$};
        
  \draw [blue, dashed, thick] (P2-R-L) -- (P2-R-M)
        node [midway, above right, yshift=1mm, fill=white, inner sep=0.5pt, font=\small] {$\Infoset{2}^{G, 1}$};
        
\end{tikzpicture}

%% file: tikz/gadget_equilibrium_maxmargin.tex
\begin{tikzpicture}[
    level 1/.style={sibling distance=4.3cm, level distance=1cm},
    level 2/.style={sibling distance=3cm, level distance=1cm}, 
    level 3/.style={sibling distance=1.4cm, level distance=1cm},
    level 4/.style={sibling distance=0.8cm, level distance=1cm},
    p2node/.style={ 
      regular polygon,
      regular polygon sides=3,
      rotate=180, 
      draw,
      fill=blue!40,
      inner sep=1pt,
      minimum size=15pt
    },
    p1node/.style={
      regular polygon,
      regular polygon sides=3, 
      draw,
      fill=red!40,
      inner sep=1pt,
      minimum size=15pt
    },
    chancenode/.style={circle, draw, fill=gray!30, minimum size=5mm, inner sep=1pt},
    terminal/.style={inner sep=2pt},
    payoff/.style={
      font=\normalsize
    },
    edge_label/.style={midway, fill=white, inner sep=1pt}
]

\node[p2node] (S) {} 
    child {
        node[chancenode] (P2-R-L) {}
        child  {
            node[p1node] (P1-L) {}
            child { 
                node[p2node] (P2-L1) {}
                child {
                    node[payoff] (T1) {$\frac{7}{4}$}
                    edge from parent
                    node[edge_label,above left] {H}
                }
                child {
                    node[payoff] (T2) {$\frac{3}{4}$}
                    edge from parent
                    node[edge_label,above right] {T}
                } 
                edge from parent
                node[edge_label,above left] {H}
            }
            child { 
                node[p2node] (P2-L2) {}
                child {
                    node[payoff] (T3) {$\frac{-1}{4}$}
                    edge from parent
                    node[edge_label,above left] {H}
                }
                child {
                    node[payoff] (T4) {$\frac{7}{4}$}
                    edge from parent
                    node[edge_label,above right] {T}
                } 
                edge from parent
                node[edge_label,above right] {T}
            }
            edge from parent
            node[edge_label,  above left] {$\frac{1}{3}$}
        }
        child {
            node[p1node] (P1-M) {}
            child { 
                node[p2node] (P2-M1) {}
                child {
                    node[payoff] (T5) {$\frac{3}{4}$}
                    edge from parent
                    node[edge_label,above left] {H}
                }
                child {
                    node[payoff] (T6) {$\frac{-1}{4}$}
                    edge from parent
                    node[edge_label,above right] {T}
                } 
                edge from parent
                node[edge_label, above left] {H}
            }
            child { 
                node[p2node] (P2-M2) {}
                child {
                    node[payoff] (T7) {$\frac{-1}{4}$}
                    edge from parent
                    node[edge_label,above left] {H}
                }
                child {
                    node[payoff] (T8) {$\frac{3}{4}$}
                    edge from parent
                    node[edge_label,above right] {T}
                } 
                edge from parent
                node[edge_label, above right] {T}
            }
            edge from parent
            node[edge_label, above right] {$\frac{1}{3}$}
        }
        edge from parent
        node[edge_label, above left] {$\Infoset{2}^1$}
    }
    child {
        node[chancenode] (P2-R-R) {}
        child {
            node[p1node] (P1-R) {}
            child { 
                node[p2node] (P2-R1) {}
                child {
                    node[payoff] (T9) {$\frac{1}{2}$}
                    edge from parent
                    node[edge_label,above left] {H}
                }
                child {
                    node[payoff] (T10) {$\frac{-1}{2}$}
                    edge from parent
                    node[edge_label,above right] {T}
                } 
                edge from parent
                node[edge_label,above left] {H}
            }
            child { 
                node[p2node] (P2-R2) {}
                child {
                    node[payoff] (T11) {$\frac{-1}{2}$}
                    edge from parent
                    node[edge_label,above left] {H}
                }
                child {
                    node[payoff] (T12) {$\frac{1}{2}$}
                    edge from parent
                    node[edge_label,above right] {T}
                } 
                edge from parent
                node[edge_label,above right] {T}
            }
            edge from parent
            node[edge_label, left] {$\frac{1}{3}$}
        }
        edge from parent
        node[edge_label, above right] {$\Infoset{2}^2$}
    }; 
  \draw [red, dashed, thick] (P1-M) -- (P1-R)
        node [midway, above left, yshift=1mm, fill=white, inner sep=1pt, font=\small] {$\Infoset{1}^2$};
        
  \draw [blue, dashed, thick] (P2-L1) -- (P2-L2)
        node [midway, above, yshift=1mm, fill=white, inner sep=1pt, font=\small] {$\Infoset{2}^1$};
  \draw [blue, dashed, thick] (P2-L2) -- (P2-M1)
        node [midway, above, yshift=1mm, fill=white, inner sep=1pt, font=\small] {$\Infoset{2}^1$};

  \draw [blue, dashed, thick] (P2-M1) -- (P2-M2)
        node [midway, above, yshift=1mm, fill=white, inner sep=1pt, font=\small] {$\Infoset{2}^1$};
        
  \draw [blue, dashed, thick] (P2-R1) -- (P2-R2)
        node [midway, above, yshift=1mm, fill=white, inner sep=1pt, font=\small] {$\Infoset{2}^2$};
         
\end{tikzpicture}

%% file: theorems/ggse_superiority_proof.tex
\begin{proof}
    Without loss of generality, we will prove this statement for player 1. We start by showing that for every Nash equilibrium, there is a gadget game sequential equilibrium that plays the same in information sets with non-zero probability, but may differ in those with zero reach probability.

    Let $\Subgame$ be a subgame, $\BlueprintStrategy{1}$ a blueprint strategy of Player 1, and $\Nash{}$ be a Nash equilibrium of the corresponding gadget game $\GadgetGame{1}$. The expected utility in the initial history of the gadget game $\InitHistory$ of the equilibrium in the gadget game is
    \begin{align} \
    \sum_{\History \in \Histories(\PublicState)} \PlayerReach{1}{\BlueprintStrategy{1}}(\History) \PlayerReach{\ChancePlayer}{}(\History) \PlayerReach{2}{\Nash{2}}(\History) \Utility{1}^{\Nash{}}(\History)
    \end{align}
    Every Nash equilibrium $\Nash{}$ is a mutual best response, this means that the following holds for Player 1 in the root information sets $\Infoset{1} \in \Infosets{1}(\PublicState)$ of the subgame $\Subgame$

    \begin{align}
        &\sum_{\History \in \Histories(\PublicState)} \PlayerReach{1}{\BlueprintStrategy{1}}(\History) \PlayerReach{\ChancePlayer}{}(\History) \PlayerReach{2}{\Nash{2}}(\History) \Utility{1}^{\Nash{}}(\History)  \\
        =&\max_{\Strategy{1}}\sum_{\History \in \Histories(\PublicState)} \PlayerReach{1}{\BlueprintStrategy{1}}(\History) \PlayerReach{\ChancePlayer}{}(\History) \PlayerReach{2}{\Nash{2}}(\History) \Utility{1}^{\Strategy{1}, \Nash{2}}(\History)  \\
        =&\max_{\Strategy{1}}\sum_{\Infoset{1} \in \Infosets{1}(\PublicState)} \PlayerReach{1}{\BlueprintStrategy{1}}(\Infoset{1}) \sum_{\Histories \in \Histories(\Infoset{1})} \PlayerReach{\ChancePlayer}{}(\History) \PlayerReach{2}{\Nash{2}}(\History) \Utility{1}^{\Strategy{1}, \Nash{2}}(\History) \\
        =&\max_{\Strategy{1}}\sum_{\Infoset{1} \in \Infosets{1}(\PublicState)} \PlayerReach{1}{\BlueprintStrategy{1}}(\Infoset{1})   \CounterfactualValue{1}{\Strategy{1}, \Nash{2}}(\Infoset{1}) \\
        =&\sum_{\Infoset{1} \in \Infosets{1}(\PublicState)} \PlayerReach{1}{\BlueprintStrategy{1}}(\Infoset{1}) \max_{\Strategy{1}} \CounterfactualValue{1}{\Strategy{1}, \Nash{2}}(\Infoset{1})
    \end{align}
    It is possible to swap the maximum and the sum in games with perfect recall.
    
    Let us split the root information sets into 2 sets $\Infosets{1}^1, \Infosets{1}^2$. The first $\Infosets{1}^1$ contains all of the information sets, that are reached with non-zero probability by the opponent $\sum_{\History \in \Histories(\Infoset{1})} \PlayerReach{2}{\Nash{2}} > 0$ and second $\Infosets{1}^2$ contains those reached with zero probability $\sum_{\History \in \Histories(\Infoset{1})} \PlayerReach{2}{\Nash{2}} = 0$.

    Now consider a GGSE with arbitrary prior $\StrategyPrior{2}$, in information sets $\Infosets{1}^1$, this prior has no effect, as $\PriorEpsilon \to 0^+$ because the reaches from the strategy outweigh the prior. This means that there is always a GGSE that plays the same in those parts. 
    
    Now we show that for each information set in $\Infoset{1}^2$, Player 1 improves the counterfactual value when converging to GGSE if Player 2 plays according to the prior. In each information set $\Infoset{1}^2 \in \Infosets{1}^2$, the reaches from the prior are $\PlayerReach{2}{\StrategyPrior{2}}(\Infoset{2}) := \StrategyPrior{2}(\AuxiliaryInfoset{2}, \Action{2}) \cdot \PriorEpsilon$ for corresponding auxiliary information set $\AuxiliaryInfoset{2}$ and action $\Action{2}$ that leads to $\Infoset{2}$. The GGSE maximizes the counterfactual values in each information set $\Infoset{1}^2$ if the other player would reach those information sets according to the prior
    \begin{align}
        \max_{\Strategy{1}} \sum_{\History \in \Histories(\Infoset{1})}\PlayerReach{\ChancePlayer}{}(\History) \PlayerReach{2}{\StrategyPrior{2}}(\History) \Utility{1}^{\Strategy{1}, \Nash{2}}(\Infoset{1}^2)
    \end{align}

    Finally, we will show that for each Nash equilibrium $\Nash{}$ and counterfactual reach $\PlayerReach{-1}{}(\History) = \PlayerReach{c}{} (\History) \PlayerReach{2}{}(\Infoset{2}(\History))$, there is a GGSE that have the same or better best response counterfactual value in each opponent's information set, assuming the same counterfactual reach. We will use $\PlayerReach{2}{}(\History) := \PlayerReach{2}{}(\Infoset{2}(\History))$ to improve clarity.

    A best response of Player 2, if we assume it has reached each information set in the root of the subgame $\Infoset{2} \in \Infosets{2}(\PublicState)$ with probability $\PlayerReach{2}{}(\Infoset{2})$ is
    \begin{align}
        &\min_{\Strategy{2}} \sum_{\History \in \Histories(\PublicState)} \PlayerReach{1}{\BlueprintStrategy{1}}(\History)\PlayerReach{\ChancePlayer}{}(\History) \PlayerReach{2}{}(\History) \Utility{1}^{\Nash{1}, \Strategy{2}} (\History) \\
        =& \min_{\Strategy{2}} \sum_{\Infoset{1} \in \Infosets{1}(\PublicState)} \PlayerReach{1}{\BlueprintStrategy{1}}(\Infoset{1}) \smashoperator{\sum_{\History \in \Histories(\Infoset{1})}}\PlayerReach{\ChancePlayer}{}(\History) \PlayerReach{2}{}(\History)\Utility{1}^{\Nash{1}, \Strategy{2}} (\History) \\
        \leq & \min_{\Strategy{2}} \sum_{\Infoset{1} \in \Infosets{1}(\PublicState)} \PlayerReach{1}{\BlueprintStrategy{1}}(\Infoset{1}) \smashoperator{\sum_{\History \in \Histories(\Infoset{1})}}\PlayerReach{\ChancePlayer}{}(\History) \PlayerReach{2}{}(\History)\Utility{1}^{\Strategy{1}^G, \Strategy{2}} (\History) \\
        =& \min_{\Strategy{2}} \sum_{\History \in \Histories(\PublicState)} \PlayerReach{1}{\BlueprintStrategy{1}}(\History)\PlayerReach{\ChancePlayer}{}(\History) \PlayerReach{2}{}(\History) \Utility{1}^{\Strategy{1}^G, \Strategy{2}} (\History)\\
        =& \min_{\Strategy{2}} \sum_{\Infoset{2} \in \Infosets{2}(\PublicState)} \PlayerReach{2}{}(\Infoset{2}) \smashoperator{\sum_{\History \in \Histories(\Infoset{2})}} \PlayerReach{1}{\BlueprintStrategy{1}}(\History)\PlayerReach{\ChancePlayer}{}(\History) \Utility{1}^{\Strategy{1}^G, \Strategy{2}} (\History) \\
        =&  \smashoperator{\sum_{\Infoset{2} \in \Infosets{2}(\PublicState)}} \PlayerReach{2}{}(\Infoset{2}) \min_{\Strategy{2}}\sum_{\History \in \Histories(\Infoset{2})} \PlayerReach{1}{\BlueprintStrategy{1}}(\History)\PlayerReach{\ChancePlayer}{}(\History) \Utility{1}^{\Strategy{1}^G, \Strategy{2}} (\History)\\
        =& \sum_{\Infoset{2} \in \Infosets{2}(\PublicState)} \PlayerReach{2}{}(\Infoset{2}) \min_{\Strategy{2}} \CounterfactualValue{1}{\BlueprintStrategy{1} \gets \Strategy{1}^G, \Strategy{2}} (\Infoset{2})
    \end{align}

    The inequality holds because in information sets from $\Infosets{1}^1$, the strategy is the same between the two equilibria, so the value will not change. In $\Infosets{1}^{2}$ the value can differ. Still, if we assume GGSE was computed with a prior corresponding to $\PlayerReach{2}{}$, then the value cannot get smaller, as the counterfactual value is maximal.

\end{proof}

%% file: theorems/non_existence_proof.tex
\begin{proof}
Consider the modification of the game Rock-Paper-Scissors, where players move sequentially. Player 2 chooses rock, paper, or scissors first, and then Player 1, unbeknownst to Player 2, chooses between rock, paper, and scissors. We show an extensive-form representation of this game in \cref{fig:non_existence_example}. The equilibrium of this game is playing each action uniformly. Assume a blueprint $\BlueprintStrategy{1}(\Infoset{1}, R) = 1$.

Since Player 2 observed the action, it can distinguish all three possible histories in $\Infoset{1}$. The counterfactual values in each history are equal to the reward given by the blueprint.

The resolving and max-margin gadget games are shown in \cref{fig:non_existence_resolving,fig:non_existence_maxmargin}. Both the resolving and max-margin gadget games have a single Nash equilibrium for Player 1, which is the same as the blueprint, e.g., $\Nash{1}(\Infoset{1}, R) = 1$. However, this strategy is fully exploitable.
\begin{figure}
    \centering
    \input{tikz/rps_example}
    \caption{An example of a game where neither gadget game can reconstruct the optimal continuation strategy. It is a known game, Rock-Paper-Scissors, in an extensive-form game where Player 2 selects rock, paper, or scissors first, then Player 1 selects as second without observing the choice of Player 2}
    \label{fig:non_existence_example}
\end{figure}
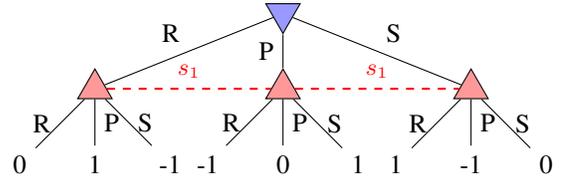

\begin{figure}
    \centering
    \input{tikz/rps_resolving}
    \caption{Resolving gadget game of rock-paper-scissors shown in \cref{fig:non_existence_example}}
    \label{fig:non_existence_resolving}
\end{figure}
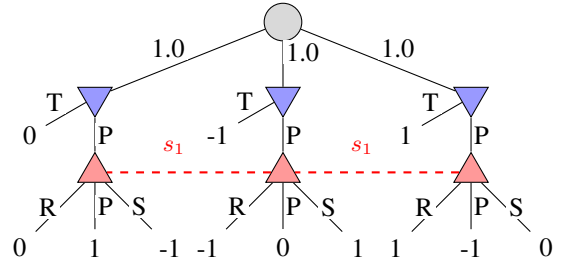

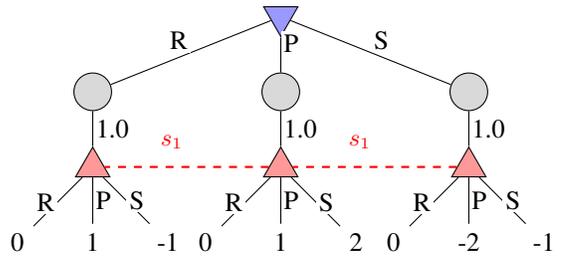
\begin{figure}
    \centering
    \input{tikz/rps_maxmargin}
    \caption{Max-margin gadget game of rock-paper-scissors shown in \cref{fig:non_existence_example}}
    \label{fig:non_existence_maxmargin}
\end{figure}
\end{proof}

%% file: tikz/rps_example.tex
\begin{tikzpicture}[
    level 1/.style={sibling distance=2.5cm, level distance=1cm},
    level 2/.style={sibling distance=1cm, level distance=1cm},
    p2node/.style={ 
      regular polygon,
      regular polygon sides=3,
      rotate=180, 
      draw,
      fill=blue!40,
      inner sep=1pt,
      minimum size=15pt
    },
    p1node/.style={
      regular polygon,
      regular polygon sides=3, 
      draw,
      fill=red!40,
      inner sep=1pt,
      minimum size=15pt
    },
    payoff/.style={
      font=\normalsize
    }
  ]

  \node [p2node] (root) {}
    child {
      node [p1node] (p1_rock) {} 
      child { node [payoff] {0} edge from parent node [left] {R} }
      child { node [payoff] {1}  edge from parent node [right] {P} } 
      child { node [payoff] {-1} edge from parent node [right] {S} }
      edge from parent node [above left] {R} 
    }
    child {
      node [p1node] (p1_paper) {} 
      child { node [payoff] {-1}  edge from parent node [left] {R} }
      child { node [payoff] {0}  edge from parent node [right] {P} } 
      child { node [payoff] {1} edge from parent node [right] {S} }
      edge from parent node [left] {P} 
    }
    child {
      node [p1node] (p1_scissors) {} 
      child { node [payoff] {1}  edge from parent node [left] {R} }
      child { node [payoff] {-1}  edge from parent node [right] {P} } 
      child { node [payoff] {0} edge from parent node [right] {S} }
      edge from parent node [above right] {S} 
    };

  \draw [red, dashed, thick] (p1_rock) -- (p1_paper)
        node [midway, above, yshift=1mm, fill=white, inner sep=1pt, font=\small] {$\Infoset{1}$};
  \draw [red, dashed, thick] (p1_paper) -- (p1_scissors)
        node [midway, above, yshift=1mm, fill=white, inner sep=1pt, font=\small] {$\Infoset{1}$};

\end{tikzpicture}

%% file: tikz/rps_resolving.tex
\begin{tikzpicture}[
    level 1/.style={sibling distance=2.5cm, level distance=1cm},
    level 2/.style={sibling distance=1cm, level distance=1cm},
    level 3/.style={sibling distance=1cm, level distance=1cm},
    p2node/.style={ 
      regular polygon,
      regular polygon sides=3,
      rotate=180, 
      draw,
      fill=blue!40,
      inner sep=1pt,
      minimum size=15pt
    },
    p1node/.style={
      regular polygon,
      regular polygon sides=3, 
      draw,
      fill=red!40,
      inner sep=1pt,
      minimum size=15pt
    },
    chancenode/.style={circle, draw, fill=gray!30, minimum size=5mm, inner sep=1pt},
    terminal/.style={inner sep=2pt},
    payoff/.style={
      font=\normalsize
    },
    edge_label/.style={midway, fill=white, inner sep=1pt}
]

\node[chancenode] (S) {} 
    child {
        node[p2node] (P2-L) {}
        child[grow=-150] { 
            node[payoff] (T1) {0}
            edge from parent
            node[edge_label, above left] {T}
        }
        child[grow=-90] { 
            node[p1node] (P1-L) {}
            child {
                node[payoff] (T2) {0}
                edge from parent
                node[edge_label, left] {R}
            }
            child {
                node[payoff] (T3) {1}
                edge from parent
                node[edge_label, right] {P}
            }
            child {
                node[payoff] (T4) {-1}
                edge from parent
                node[edge_label, right] {S}
            }
            edge from parent
            node[edge_label, right] {P}
        }
        edge from parent
        node[edge_label, above left] {1.0}
    }
    child {
        node[p2node] (P2-R) {}
        child[grow=-150] { 
            node[payoff] (T5) {-1}
            edge from parent
            node[edge_label, above left] {T}
        }
        child[grow=-90] { 
            node[p1node] (P1-M) {}
            child {
                node[payoff] (T6) {-1}
                edge from parent
                node[edge_label, left] {R}
            }
            child {
                node[payoff] (T7) {0}
                edge from parent
                node[edge_label, right] {P}
            }
            child {
                node[payoff] (T8) {1}
                edge from parent
                node[edge_label, right] {S}
            }
            edge from parent
            node[edge_label, right] {P}
        }
        edge from parent
        node[edge_label, above right] {1.0}
    }
    child {
        node[p2node] (P2-R) {}
        child[grow=-150] { 
            node[payoff] (T9) {1}
            edge from parent
            node[edge_label, above left] {T}
        }
        child[grow=-90] { 
            node[p1node] (P1-R) {}
            child {
                node[payoff] (T10) {1}
                edge from parent
                node[edge_label, left] {R}
            }
            child {
                node[payoff] (T11) {-1}
                edge from parent
                node[edge_label, right] {P}
            }
            child {
                node[payoff] (T12) {0}
                edge from parent
                node[edge_label, right] {S}
            }
            edge from parent
            node[edge_label, right] {P}
        }
        edge from parent
        node[edge_label, above right] {1.0}
    };
  \draw [red, dashed, thick] (P1-R) -- (P1-M)
        node [midway, above left, yshift=2mm, fill=white, inner sep=1pt, font=\small] {$\Infoset{1}$};
  \draw [red, dashed, thick] (P1-M) -- (P1-L)
        node [midway, above left, yshift=2mm, fill=white, inner sep=1pt, font=\small] {$\Infoset{1}$};

\end{tikzpicture}

%% file: tikz/rps_maxmargin.tex
\begin{tikzpicture}[
    level 1/.style={sibling distance=2.5cm, level distance=1cm},
    level 2/.style={sibling distance=1cm, level distance=1cm},
    level 3/.style={sibling distance=1cm, level distance=1cm},
    p2node/.style={ 
      regular polygon,
      regular polygon sides=3,
      rotate=180, 
      draw,
      fill=blue!40,
      inner sep=1pt,
      minimum size=15pt
    },
    p1node/.style={
      regular polygon,
      regular polygon sides=3, 
      draw,
      fill=red!40,
      inner sep=1pt,
      minimum size=15pt
    },
    chancenode/.style={circle, draw, fill=gray!30, minimum size=5mm, inner sep=1pt},
    terminal/.style={inner sep=2pt},
    payoff/.style={
      font=\normalsize
    },
    edge_label/.style={midway, fill=white, inner sep=1pt}
]

\node[p2node] (S) {} 
    child {
        node[chancenode] (P2-L) {}
        child { 
            node[p1node] (P1-L) {}
            child {
                node[payoff] (T2) {0}
                edge from parent
                node[edge_label, left] {R}
            }
            child {
                node[payoff] (T3) {1}
                edge from parent
                node[edge_label, right] {P}
            }
            child {
                node[payoff] (T4) {-1}
                edge from parent
                node[edge_label, right] {S}
            }
            edge from parent
            node[edge_label, right] {1.0}
        }
        edge from parent
        node[edge_label, above left] {R}
    }
    child {
        node[chancenode] (P2-M) {}
        child{ 
            node[p1node] (P1-M) {}
            child {
                node[payoff] (T6) {0}
                edge from parent
                node[edge_label, left] {R}
            }
            child {
                node[payoff] (T7) {1}
                edge from parent
                node[edge_label, right] {P}
            }
            child {
                node[payoff] (T8) {2}
                edge from parent
                node[edge_label, right] {S}
            }
            edge from parent
            node[edge_label, right] {1.0}
        }
        edge from parent
        node[edge_label, above right] {P}
    }
    child {
        node[chancenode] (P2-R) {}
        child { 
            node[p1node] (P1-R) {}
            child {
                node[payoff] (T10) {0}
                edge from parent
                node[edge_label, left] {R}
            }
            child {
                node[payoff] (T11) {-2}
                edge from parent
                node[edge_label, right] {P}
            }
            child {
                node[payoff] (T12) {-1}
                edge from parent
                node[edge_label, right] {S}
            }
            edge from parent
            node[edge_label, right] {1.0}
        }
        edge from parent
        node[edge_label, above right] {S}
    };
  \draw [red, dashed, thick] (P1-R) -- (P1-M)
        node [midway, above left, yshift=2mm, fill=white, inner sep=1pt, font=\small] {$\Infoset{1}$};
  \draw [red, dashed, thick] (P1-M) -- (P1-L)
        node [midway, above left, yshift=2mm, fill=white, inner sep=1pt, font=\small] {$\Infoset{1}$};

\end{tikzpicture}

%% file: theorems/maxmargin_worse_proof.tex
\begin{proof}

    Assume a game that starts with a chance node that uniformly decides one of 3 states. Player 1 cannot distinguish between the second and third states, while Player 2 cannot distinguish between the first and second states. Each of the states is followed by a matching pennies game with modified utilities. In the first state, Player 1 gets a reward of 2 if both players play heads; in the second state, Player 1 gets a reward of 3 if both play tails; in the third state, it is regular matching pennies. The sequence-form representation of this game is in \cref{fig:maxmargin_worse_example}.

    This game has a single Nash equilibrium, where Player 1 plays $\Nash{1}(\Infoset{1}^{1}, H) = 0$, $\Nash{1}(\Infoset{1}^2, H) = \frac{1}{4}$
    Assume following blueprint strategy of Player 1 $\BlueprintStrategy{1}(\Infoset{1}^{1}, H) = 1.0$, $\BlueprintStrategy{1}(\Infoset{1}^2, H) = \frac{1}{4}$.

    The Player 2 counterfactual best response values against this blueprint are $\CounterfactualValue{1}{\BlueprintStrategy{-2}}
    (\Infoset{2}^1) = \frac{1}{8}$, $\CounterfactualValue{1}{\BlueprintStrategy{-2}}(\Infoset{2}^2) = \frac{1}{4}$.

    The corresponding resolving and max-margin gadget games are in \cref{fig:maxmargin_worse_resolving,fig:maxmargin_worse_maxmargin}. In every Nash equilibrium of the max-margin gadget game, Player 2 plays with probability 1 in the information set $\Infoset{2}^2$. Since the right part of the game tree is just Matching pennies with shifted utilities, Player 1's part of Nash will always play $\Strategy{1}(\Infoset{1}^2, H) = \frac{1}{2}$. However, no equilibrium in the original game would play such a strategy. This means that there is no Nash equilibrium in the max-margin gadget game, which would be the Nash equilibrium in the original game.

    On the other hand, any strategy in the resolving gadget game, which ensures greater counterfactual values in the subgame than the blueprint ones, is a Nash equilibrium. In this case, Player 1's part of the Nash equilibrium in the original game satisfies these bounds, so it is also a Nash equilibrium of the resolving gadget game. This means that the resolving gadget game can reconstruct the Nash equilibrium of the original game, while max-margin cannot.
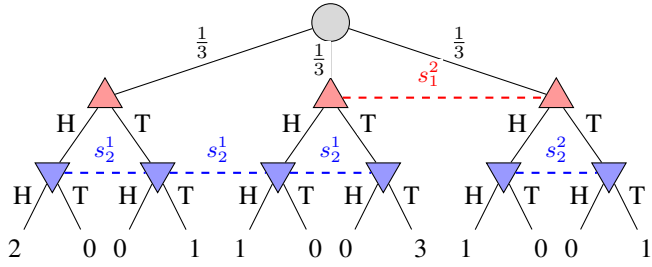
\begin{figure}
    \centering
    \input{tikz/maxmargin_worse_example}
    \caption{An example of a game, where neither Nash equilibrium of the max-margin gadget game outperforms some equilibrium in resolving gadget game}
    \label{fig:maxmargin_worse_example}
\end{figure}
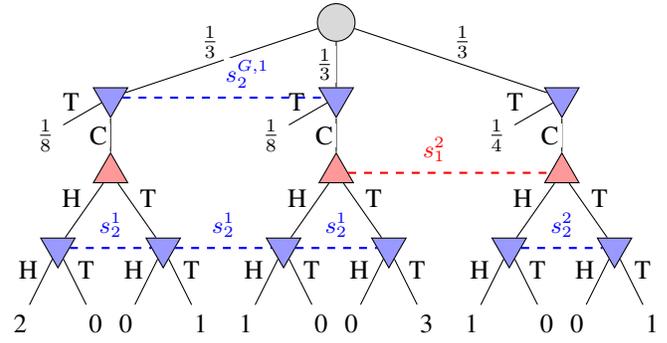
\begin{figure}
    \centering
    \input{tikz/maxmargin_worse_resolving}
    \caption{Resolving gadget game of subgame from \cref{fig:maxmargin_worse_example}}
    \label{fig:maxmargin_worse_resolving}
\end{figure}
\begin{figure}
    \centering
    \input{tikz/maxmargin_worse_maxmargin}
    \caption{Max-margin gadget game of subgame from \cref{fig:maxmargin_worse_example}}
    \label{fig:maxmargin_worse_maxmargin}
\end{figure}
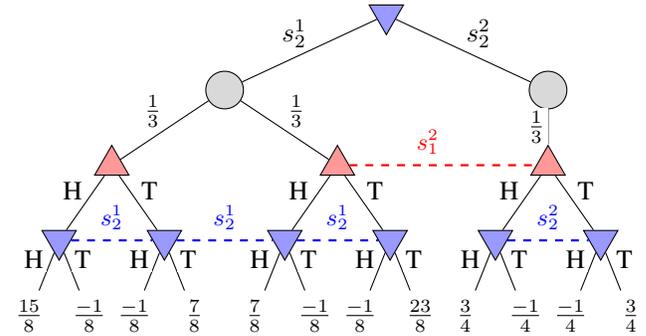
\end{proof}

%% file: tikz/maxmargin_worse_example.tex
\begin{tikzpicture}[
    level 1/.style={sibling distance=3cm, level distance=1cm}, 
    level 2/.style={sibling distance=1.4cm, level distance=1cm},
    level 3/.style={sibling distance=1cm, level distance=1cm},
    p2node/.style={ 
      regular polygon,
      regular polygon sides=3,
      rotate=180, 
      draw,
      fill=blue!40,
      inner sep=1pt,
      minimum size=15pt
    },
    p1node/.style={
      regular polygon,
      regular polygon sides=3, 
      draw,
      fill=red!40,
      inner sep=1pt,
      minimum size=15pt
    },
    chancenode/.style={circle, draw, fill=gray!30, minimum size=5mm, inner sep=1pt},
    terminal/.style={inner sep=2pt},
    payoff/.style={
      font=\normalsize
    },
    edge_label/.style={midway, fill=white, inner sep=1pt}
]

\node[chancenode] (S) {} 
    child {
        node[p1node] (P1-L) {}
        child { 
            node[p2node] (P2-L1) {}
            child {
                node[payoff] (T1) {2}
                edge from parent
                node[edge_label,above left] {H}
            }
            child {
                node[payoff] (T2) {0}
                edge from parent
                node[edge_label,above right] {T}
            } 
            edge from parent
            node[edge_label,above left] {H}
        }
        child { 
            node[p2node] (P2-L2) {}
            child {
                node[payoff] (T3) {0}
                edge from parent
                node[edge_label,above left] {H}
            }
            child {
                node[payoff] (T4) {1}
                edge from parent
                node[edge_label,above right] {T}
            } 
            edge from parent
            node[edge_label,above right] {T}
        }
        edge from parent
        node[edge_label, above left] {$\frac{1}{3}$}
    }
    child {
        node[p1node] (P1-M) {}
        child { 
            node[p2node] (P2-M1) {}
            child {
                node[payoff] (T5) {1}
                edge from parent
                node[edge_label,above left] {H}
            }
            child {
                node[payoff] (T6) {0}
                edge from parent
                node[edge_label,above right] {T}
            } 
            edge from parent
            node[edge_label, above left] {H}
        }
        child { 
            node[p2node] (P2-M2) {}
            child {
                node[payoff] (T7) {0}
                edge from parent
                node[edge_label,above left] {H}
            }
            child {
                node[payoff] (T8) {3}
                edge from parent
                node[edge_label,above right] {T}
            } 
            edge from parent
            node[edge_label, above right] {T}
        }
        edge from parent
        node[edge_label, left] {$\frac{1}{3}$}
    }
    child {
        node[p1node] (P1-R) {}
        child { 
            node[p2node] (P2-R1) {}
            child {
                node[payoff] (T9) {1}
                edge from parent
                node[edge_label,above left] {H}
            }
            child {
                node[payoff] (T10) {0}
                edge from parent
                node[edge_label,above right] {T}
            } 
            edge from parent
            node[edge_label,above left] {H}
        }
        child { 
            node[p2node] (P2-R2) {}
            child {
                node[payoff] (T11) {0}
                edge from parent
                node[edge_label,above left] {H}
            }
            child {
                node[payoff] (T12) {1}
                edge from parent
                node[edge_label,above right] {T}
            } 
            edge from parent
            node[edge_label,above right] {T}
        }
        edge from parent
        node[edge_label, above right] {$\frac{1}{3}$}
    }; 
  \draw [red, dashed, thick] (P1-M) -- (P1-R)
        node [midway, above left, yshift=1mm, fill=white, inner sep=1pt, font=\small] {$\Infoset{1}^2$};
        
  \draw [blue, dashed, thick] (P2-L1) -- (P2-L2)
        node [midway, above, yshift=1mm, fill=white, inner sep=1pt, font=\small] {$\Infoset{2}^1$};
  \draw [blue, dashed, thick] (P2-L2) -- (P2-M1)
        node [midway, above, yshift=1mm, fill=white, inner sep=1pt, font=\small] {$\Infoset{2}^1$};

  \draw [blue, dashed, thick] (P2-M1) -- (P2-M2)
        node [midway, above, yshift=1mm, fill=white, inner sep=1pt, font=\small] {$\Infoset{2}^1$};
        
  \draw [blue, dashed, thick] (P2-R1) -- (P2-R2)
        node [midway, above, yshift=1mm, fill=white, inner sep=1pt, font=\small] {$\Infoset{2}^2$};
        
\end{tikzpicture}

%% file: tikz/maxmargin_worse_resolving.tex
\begin{tikzpicture}[
    level 1/.style={sibling distance=3cm, level distance=1cm},
    level 2/.style={sibling distance=1cm, level distance=1cm}, 
    level 3/.style={sibling distance=1.4cm, level distance=1cm},
    level 4/.style={sibling distance=1cm, level distance=1cm},
    p2node/.style={ 
      regular polygon,
      regular polygon sides=3,
      rotate=180, 
      draw,
      fill=blue!40,
      inner sep=1pt,
      minimum size=15pt
    },
    p1node/.style={
      regular polygon,
      regular polygon sides=3, 
      draw,
      fill=red!40,
      inner sep=1pt,
      minimum size=15pt
    },
    chancenode/.style={circle, draw, fill=gray!30, minimum size=5mm, inner sep=1pt},
    terminal/.style={inner sep=2pt},
    payoff/.style={
      font=\normalsize
    },
    edge_label/.style={midway, fill=white, inner sep=1pt}
]

\node[chancenode] (S) {} 
    child {
        node[p2node] (P2-R-L) {}
        child[grow=-150]  { 
            node[payoff] (TER1) {$\frac{1}{8}$}
            edge from parent
            node[edge_label, above left] {T}
        }
        child[grow=-90]  {
            node[p1node] (P1-L) {}
            child { 
                node[p2node] (P2-L1) {}
                child {
                    node[payoff] (T1) {2}
                    edge from parent
                    node[edge_label,above left] {H}
                }
                child {
                    node[payoff] (T2) {0}
                    edge from parent
                    node[edge_label,above right] {T}
                } 
                edge from parent
                node[edge_label,above left] {H}
            }
            child { 
                node[p2node] (P2-L2) {}
                child {
                    node[payoff] (T3) {0}
                    edge from parent
                    node[edge_label,above left] {H}
                }
                child {
                    node[payoff] (T4) {1}
                    edge from parent
                    node[edge_label,above right] {T}
                } 
                edge from parent
                node[edge_label,above right] {T}
            }
            edge from parent
            node[edge_label, left] {C}
        }
        edge from parent
        node[edge_label, above left] {$\frac{1}{3}$}
    }
    child {
        node[p2node] (P2-R-M) {}
        child[grow=-150]  { 
            node[payoff] (TER2) {$\frac{1}{8}$}
            edge from parent
            node[edge_label, above left] {T}
        }
        child[grow=-90] {
            node[p1node] (P1-M) {}
            child { 
                node[p2node] (P2-M1) {}
                child {
                    node[payoff] (T5) {1}
                    edge from parent
                    node[edge_label,above left] {H}
                }
                child {
                    node[payoff] (T6) {0}
                    edge from parent
                    node[edge_label,above right] {T}
                } 
                edge from parent
                node[edge_label, above left] {H}
            }
            child { 
                node[p2node] (P2-M2) {}
                child {
                    node[payoff] (T7) {0}
                    edge from parent
                    node[edge_label,above left] {H}
                }
                child {
                    node[payoff] (T8) {3}
                    edge from parent
                    node[edge_label,above right] {T}
                } 
                edge from parent
                node[edge_label, above right] {T}
            }
            edge from parent
            node[edge_label, left] {C}
        }
        edge from parent
        node[edge_label, left] {$\frac{1}{3}$}
    }
    child {
        node[p2node] (P2-R-R) {}
        child[grow=-150]  { 
            node[payoff] (TER3) {$\frac{1}{4}$}
            edge from parent
            node[edge_label, above left,xshift=-1mm] {T}
        }
        child[grow=-90] {
            node[p1node] (P1-R) {}
            child { 
                node[p2node] (P2-R1) {}
                child {
                    node[payoff] (T9) {1}
                    edge from parent
                    node[edge_label,above left] {H}
                }
                child {
                    node[payoff] (T10) {0}
                    edge from parent
                    node[edge_label,above right] {T}
                } 
                edge from parent
                node[edge_label,above left] {H}
            }
            child { 
                node[p2node] (P2-R2) {}
                child {
                    node[payoff] (T11) {0}
                    edge from parent
                    node[edge_label,above left] {H}
                }
                child {
                    node[payoff] (T12) {1}
                    edge from parent
                    node[edge_label,above right] {T}
                } 
                edge from parent
                node[edge_label,above right] {T}
            }
            edge from parent
            node[edge_label, left] {C}
        }
        edge from parent
        node[edge_label, above right] {$\frac{1}{3}$}
    }; 
  \draw [red, dashed, thick] (P1-M) -- (P1-R)
        node [midway, above left, yshift=1mm, fill=white, inner sep=1pt, font=\small] {$\Infoset{1}^2$};
        
  \draw [blue, dashed, thick] (P2-L1) -- (P2-L2)
        node [midway, above, yshift=1mm, fill=white, inner sep=1pt, font=\small] {$\Infoset{2}^1$};
  \draw [blue, dashed, thick] (P2-L2) -- (P2-M1)
        node [midway, above, yshift=1mm, fill=white, inner sep=1pt, font=\small] {$\Infoset{2}^1$};

  \draw [blue, dashed, thick] (P2-M1) -- (P2-M2)
        node [midway, above, yshift=1mm, fill=white, inner sep=1pt, font=\small] {$\Infoset{2}^1$};
        
  \draw [blue, dashed, thick] (P2-R1) -- (P2-R2)
        node [midway, above, yshift=1mm, fill=white, inner sep=1pt, font=\small] {$\Infoset{2}^2$};
        
  \draw [blue, dashed, thick] (P2-R-L) -- (P2-R-M)
        node [midway, above right, yshift=1mm, fill=white, inner sep=0.5pt, font=\small] {$\Infoset{2}^{G, 1}$};
        
\end{tikzpicture}

%% file: tikz/maxmargin_worse_maxmargin.tex
\begin{tikzpicture}[
    level 1/.style={sibling distance=4.3cm, level distance=1cm},
    level 2/.style={sibling distance=3cm, level distance=1cm}, 
    level 3/.style={sibling distance=1.4cm, level distance=1cm},
    level 4/.style={sibling distance=0.8cm, level distance=1cm},
    p2node/.style={ 
      regular polygon,
      regular polygon sides=3,
      rotate=180, 
      draw,
      fill=blue!40,
      inner sep=1pt,
      minimum size=15pt
    },
    p1node/.style={
      regular polygon,
      regular polygon sides=3, 
      draw,
      fill=red!40,
      inner sep=1pt,
      minimum size=15pt
    },
    chancenode/.style={circle, draw, fill=gray!30, minimum size=5mm, inner sep=1pt},
    terminal/.style={inner sep=2pt},
    payoff/.style={
      font=\normalsize
    },
    edge_label/.style={midway, fill=white, inner sep=1pt}
]

\node[p2node] (S) {} 
    child {
        node[chancenode] (P2-R-L) {}
        child  {
            node[p1node] (P1-L) {}
            child { 
                node[p2node] (P2-L1) {}
                child {
                    node[payoff] (T1) {$\frac{15}{8}$}
                    edge from parent
                    node[edge_label,above left] {H}
                }
                child {
                    node[payoff] (T2) {$\frac{-1}{8}$}
                    edge from parent
                    node[edge_label,above right] {T}
                } 
                edge from parent
                node[edge_label,above left] {H}
            }
            child { 
                node[p2node] (P2-L2) {}
                child {
                    node[payoff] (T3) {$\frac{-1}{8}$}
                    edge from parent
                    node[edge_label,above left] {H}
                }
                child {
                    node[payoff] (T4) {$\frac{7}{8}$}
                    edge from parent
                    node[edge_label,above right] {T}
                } 
                edge from parent
                node[edge_label,above right] {T}
            }
            edge from parent
            node[edge_label,  above left] {$\frac{1}{3}$}
        }
        child {
            node[p1node] (P1-M) {}
            child { 
                node[p2node] (P2-M1) {}
                child {
                    node[payoff] (T5) {$\frac{7}{8}$}
                    edge from parent
                    node[edge_label,above left] {H}
                }
                child {
                    node[payoff] (T6) {$\frac{-1}{8}$}
                    edge from parent
                    node[edge_label,above right] {T}
                } 
                edge from parent
                node[edge_label, above left] {H}
            }
            child { 
                node[p2node] (P2-M2) {}
                child {
                    node[payoff] (T7) {$\frac{-1}{8}$}
                    edge from parent
                    node[edge_label,above left] {H}
                }
                child {
                    node[payoff] (T8) {$\frac{23}{8}$}
                    edge from parent
                    node[edge_label,above right] {T}
                } 
                edge from parent
                node[edge_label, above right] {T}
            }
            edge from parent
            node[edge_label, above right] {$\frac{1}{3}$}
        }
        edge from parent
        node[edge_label, above left] {$\Infoset{2}^1$}
    }
    child {
        node[chancenode] (P2-R-R) {}
        child {
            node[p1node] (P1-R) {}
            child { 
                node[p2node] (P2-R1) {}
                child {
                    node[payoff] (T9) {$\frac{3}{4}$}
                    edge from parent
                    node[edge_label,above left] {H}
                }
                child {
                    node[payoff] (T10) {$\frac{-1}{4}$}
                    edge from parent
                    node[edge_label,above right] {T}
                } 
                edge from parent
                node[edge_label,above left] {H}
            }
            child { 
                node[p2node] (P2-R2) {}
                child {
                    node[payoff] (T11) {$\frac{-1}{4}$}
                    edge from parent
                    node[edge_label,above left] {H}
                }
                child {
                    node[payoff] (T12) {$\frac{3}{4}$}
                    edge from parent
                    node[edge_label,above right] {T}
                } 
                edge from parent
                node[edge_label,above right] {T}
            }
            edge from parent
            node[edge_label, left] {$\frac{1}{3}$}
        }
        edge from parent
        node[edge_label, above right] {$\Infoset{2}^2$}
    }; 
  \draw [red, dashed, thick] (P1-M) -- (P1-R)
        node [midway, above left, yshift=1mm, fill=white, inner sep=1pt, font=\small] {$\Infoset{1}^2$};
        
  \draw [blue, dashed, thick] (P2-L1) -- (P2-L2)
        node [midway, above, yshift=1mm, fill=white, inner sep=1pt, font=\small] {$\Infoset{2}^1$};
  \draw [blue, dashed, thick] (P2-L2) -- (P2-M1)
        node [midway, above, yshift=1mm, fill=white, inner sep=1pt, font=\small] {$\Infoset{2}^1$};

  \draw [blue, dashed, thick] (P2-M1) -- (P2-M2)
        node [midway, above, yshift=1mm, fill=white, inner sep=1pt, font=\small] {$\Infoset{2}^1$};
        
  \draw [blue, dashed, thick] (P2-R1) -- (P2-R2)
        node [midway, above, yshift=1mm, fill=white, inner sep=1pt, font=\small] {$\Infoset{2}^2$};
         
\end{tikzpicture}

%% file: chapters/add_experiments.tex
\section{Additional Experimental Details}
\label{app:experiments}
We provide some additional experimental details and results that were omitted from the main body of the paper.
\subsection{Statistical Test of Exploitability}
We have computed a t-test for all pairs of techniques used in \cref{sec:experiments} to support our claims. These results are in \cref{fig:ttest_goof5_d1,fig:ttest_goof5_d2,fig:ttest_leduc,fig:ttest_liars_d1}. High values indicate that the row technique is more likely to outperform the column technique.

We also computed the difference in exploitability between the unsafe resolving and resolving gadget game when using the blueprint as prior. The results in \cref{fig:diff_expl} confirm that using a resolving gadget with a blueprint prior always performs at least as well as unsafe, and is guaranteed to be safe.
\begin{figure}[]
    \centering  
    \includegraphics[width=\linewidth]{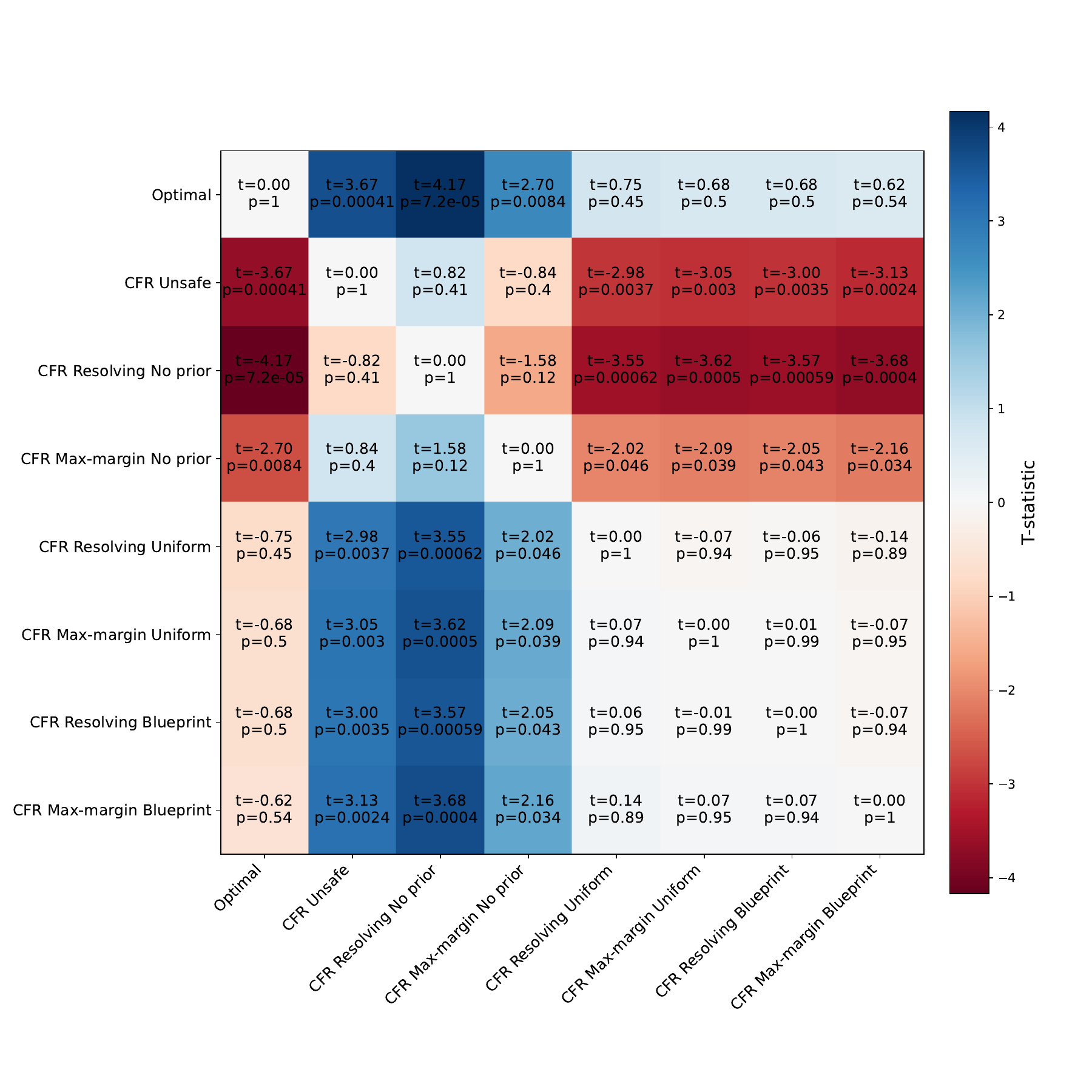}  
    \caption{The results of t-test for each combination of subgame solving techniques in Goofspiel 5 used in depth 1.} 
    \label{fig:ttest_goof5_d1} 
\end{figure}
\begin{figure}[]
    \centering  
    \includegraphics[width=\linewidth]{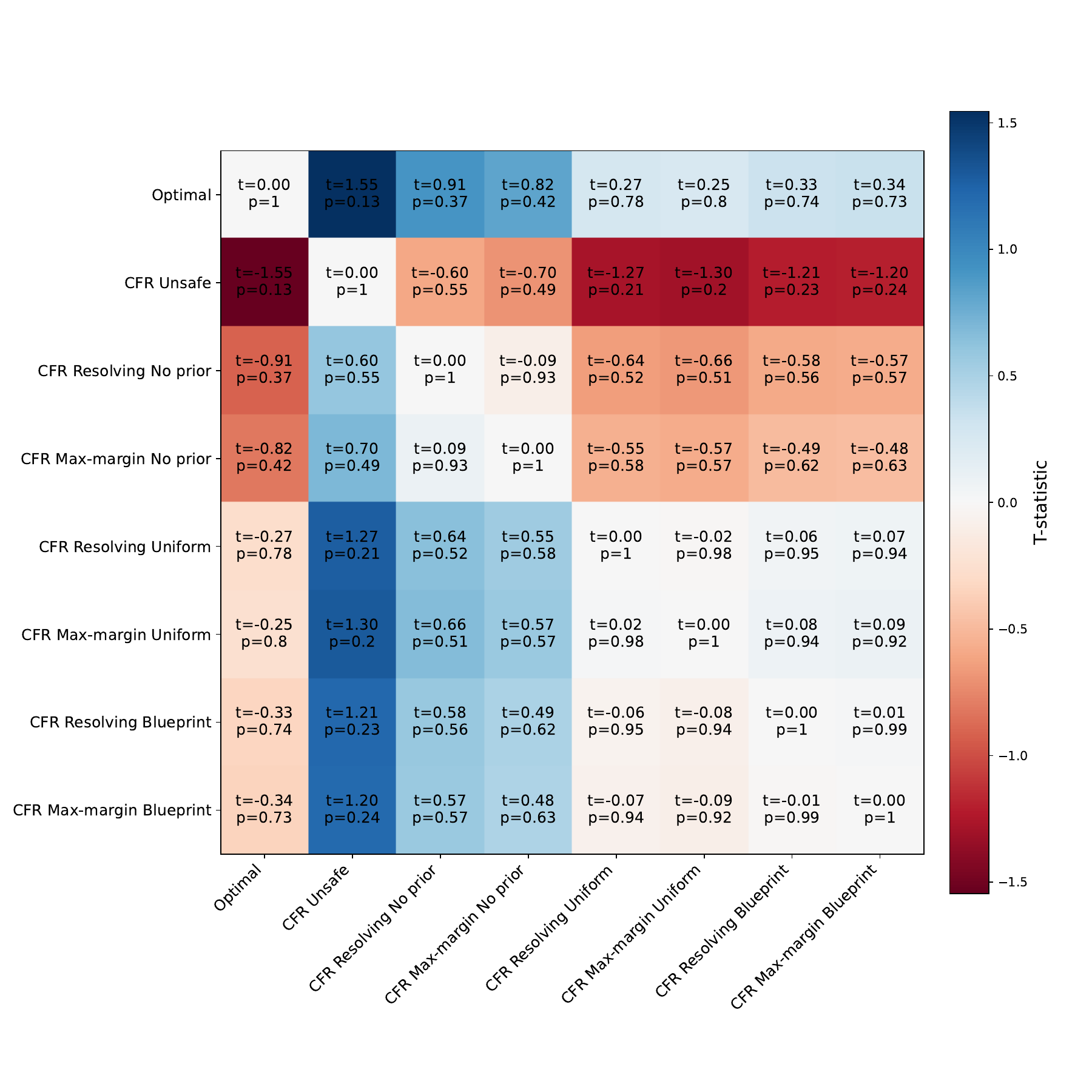}  
    \caption{The results of t-test for each combination of subgame solving techniques in Goofspiel 5 used in depth 2.} 
    \label{fig:ttest_goof5_d2} 
\end{figure}
\begin{figure}[]
    \centering  
    \includegraphics[width=\linewidth]{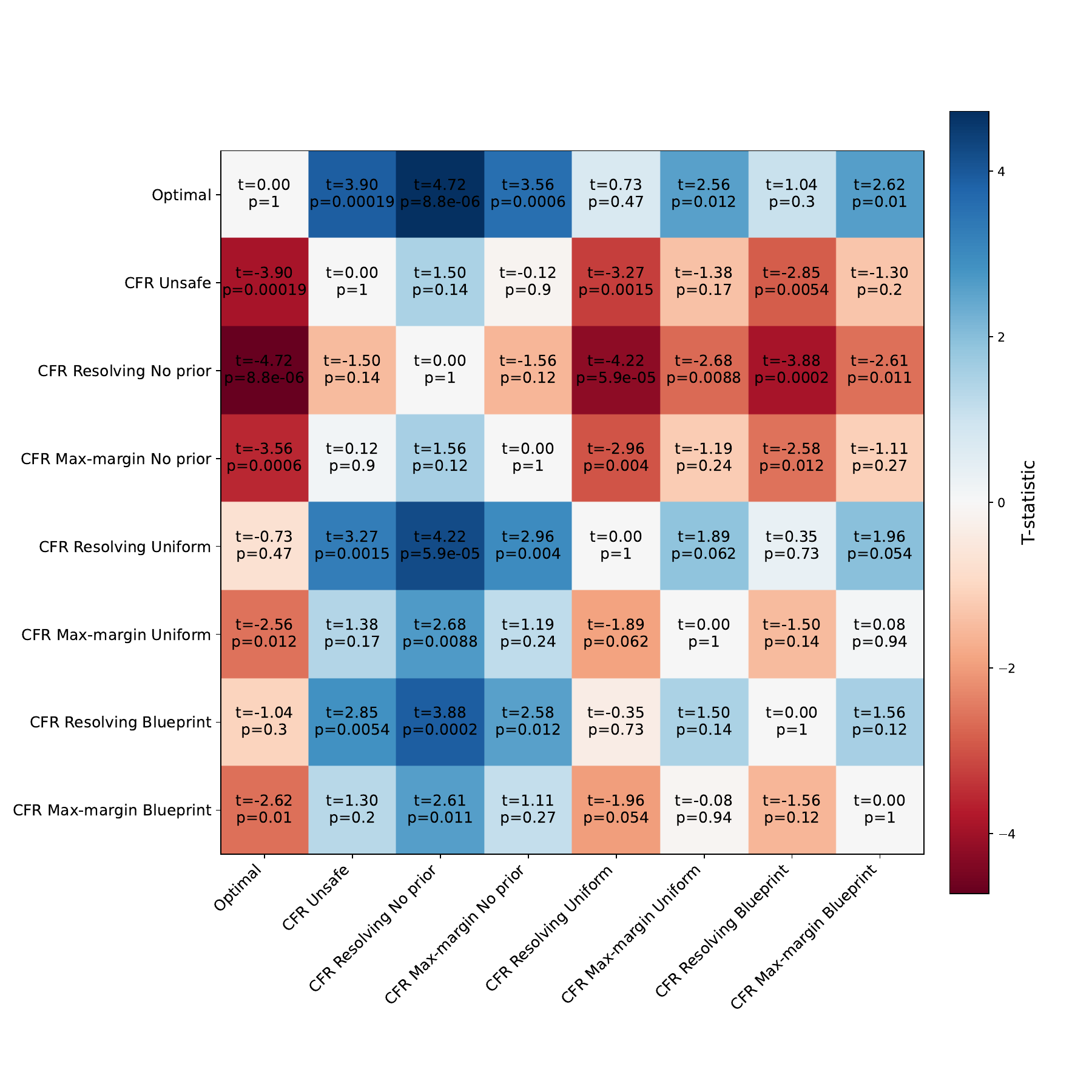}  
    \caption{The results of t-test for each combination of subgame solving techniques in Leduc after flop.} 
    \label{fig:ttest_leduc} 
\end{figure}
\begin{figure}[]
    \centering  
    \includegraphics[width=\linewidth]{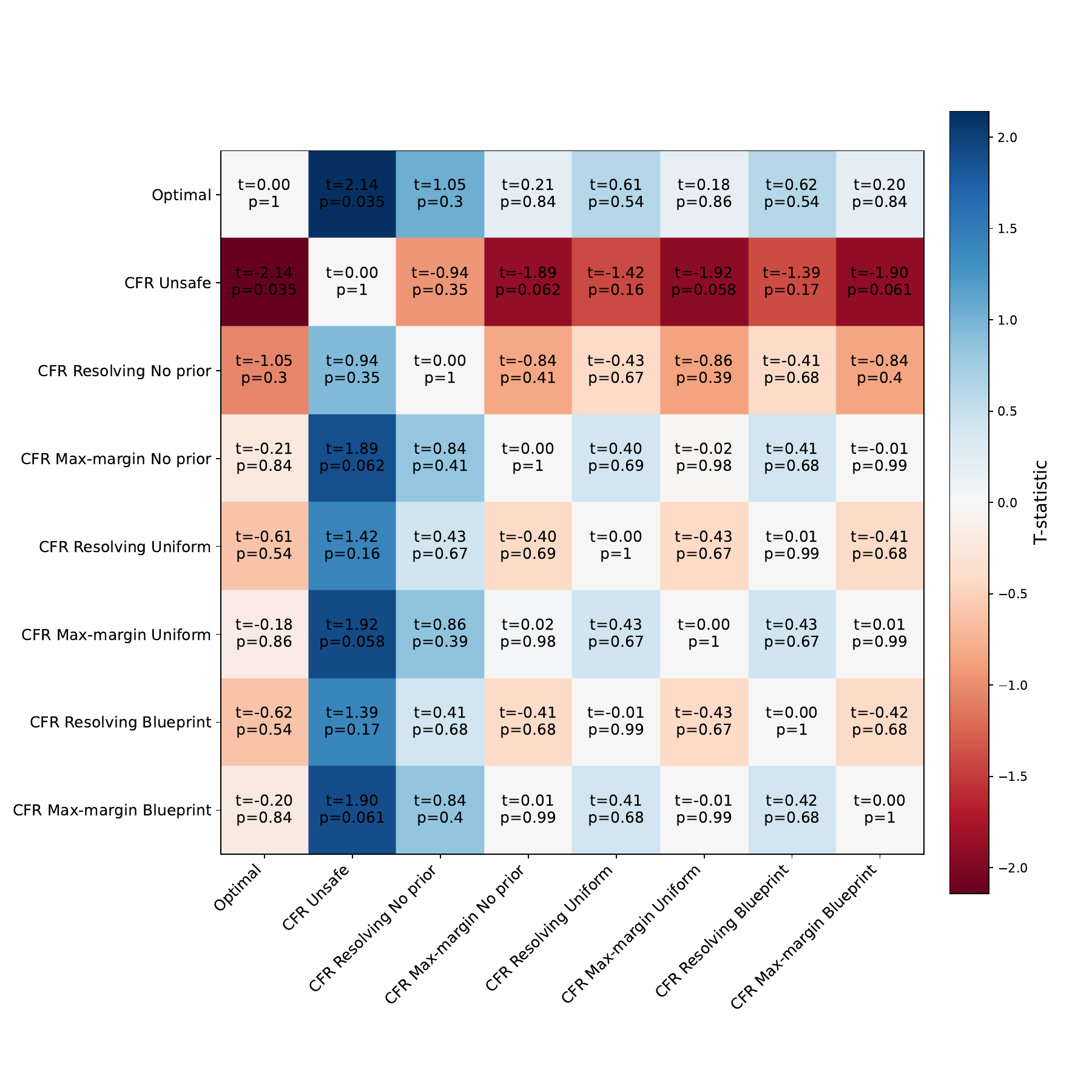}  
    \caption{The results of t-test for each combination of subgame solving techniques in Liar's dice 1,4 used in depth 1.} 
    \label{fig:ttest_liars_d1} 
\end{figure}
\begin{figure}[]
    \centering
        \begin{subfigure}[b]{0.24\textwidth}
            \centering
            \includegraphics[width=\textwidth]{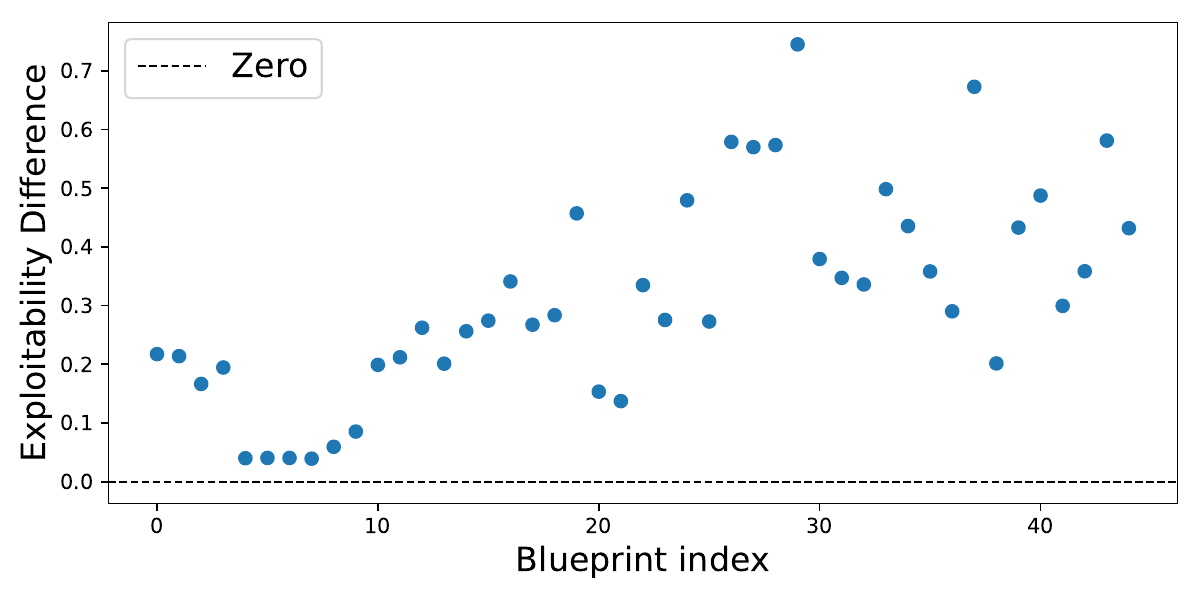} 
            \caption{Goofspiel 5, depth 1}
            \label{fig:diff_expl_goof5_d1}
        \end{subfigure}%
        \hfill
        \begin{subfigure}[b]{0.24\textwidth}
            \centering
            \includegraphics[width=\textwidth]{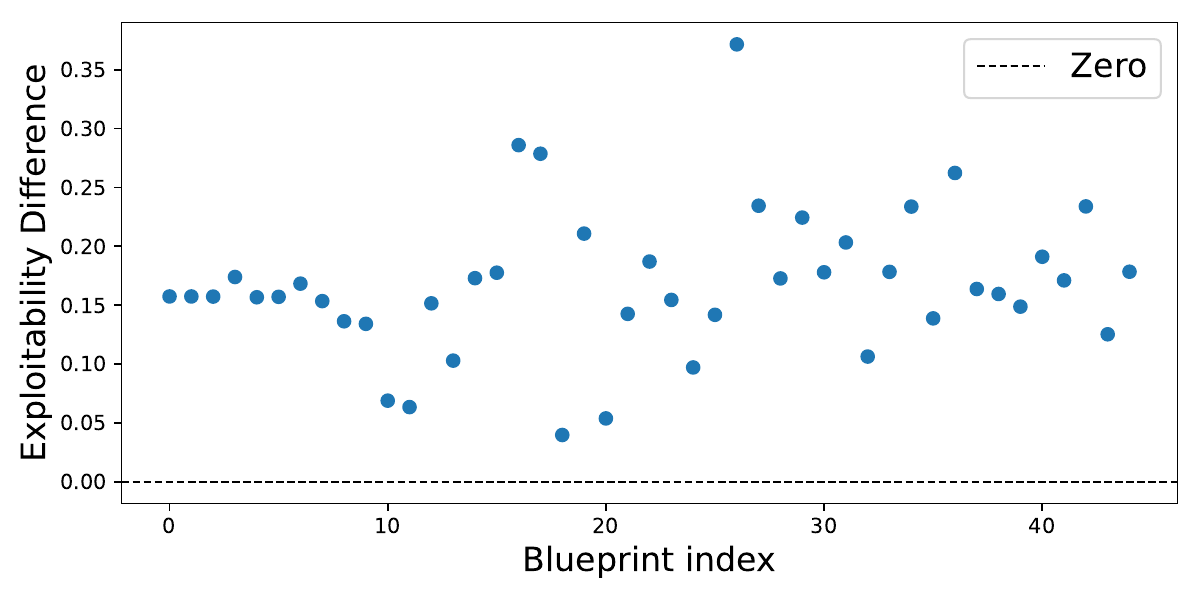} 
            \caption{Goofspiel 5, depth 2}
            \label{fig:diff_expl_goof5_d2}
        \end{subfigure}
        \begin{subfigure}[b]{0.24\textwidth}
            \centering
            \includegraphics[width=\textwidth]{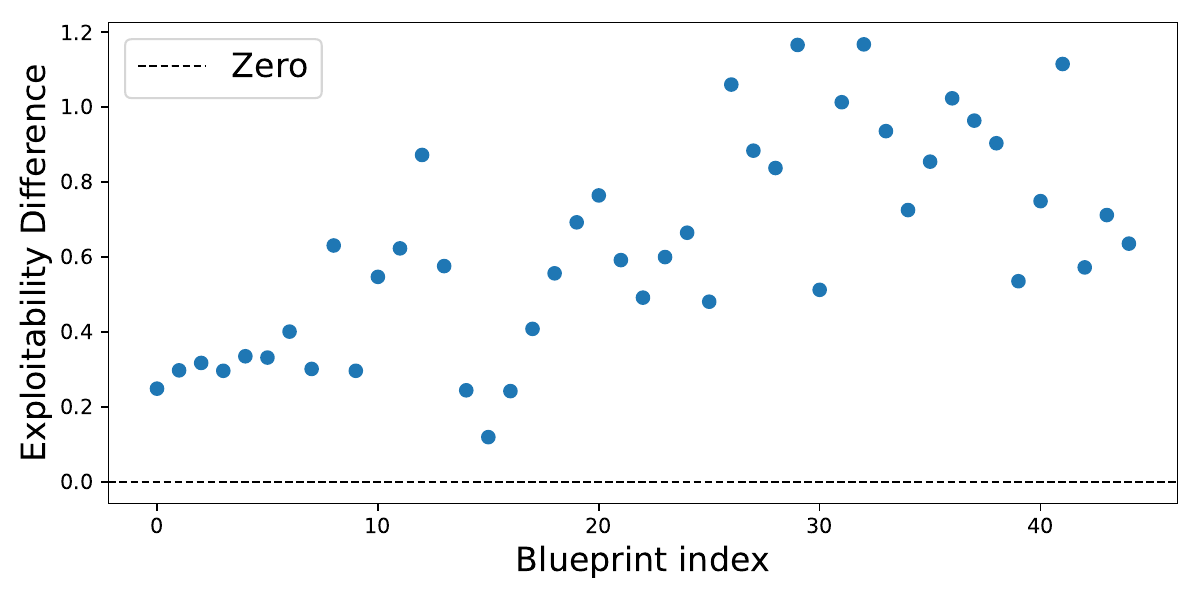} 
            \caption{Leduc hold'em after flop}
            \label{fig:diff_expl_leduc}
        \end{subfigure}%
        \hfill
        \begin{subfigure}[b]{0.24\textwidth}
            \centering
            \includegraphics[width=\textwidth]{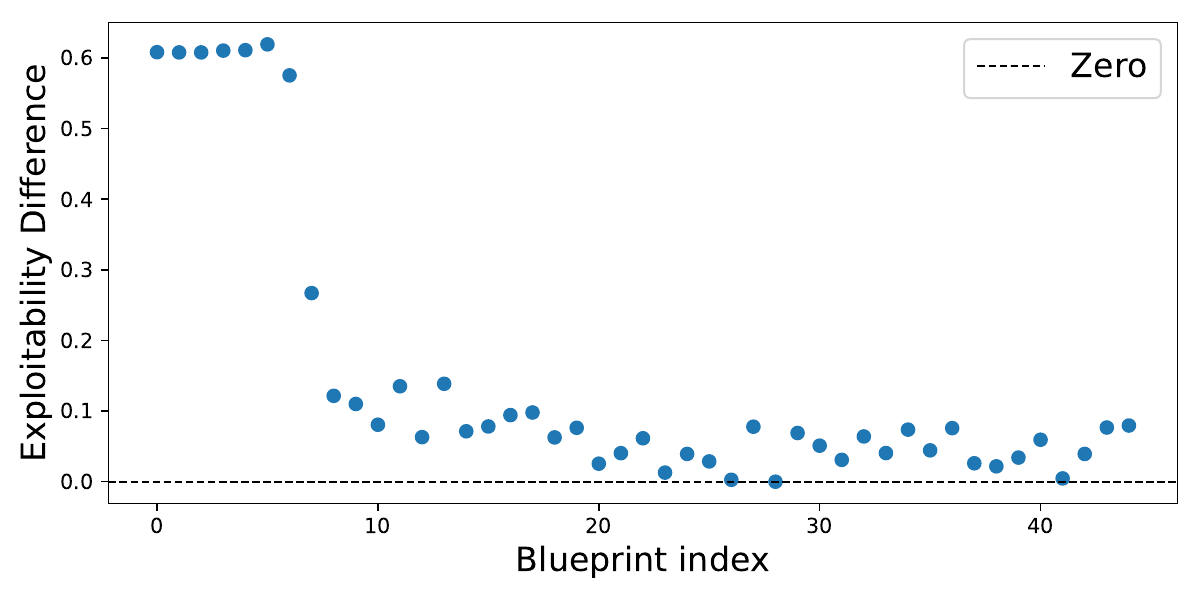} 
            \caption{Liar's dice, depth 1}
            \label{fig:diff_expl_liars_d1}
        \end{subfigure}
    \caption{Difference between exploitability of the strategy computed by unsafe solving and resolving gadget game with blueprint prior. The points are ordered based on the exploitability of the blueprint from left to right.}
    \label{fig:diff_expl}
\end{figure}
\subsection{Sequence-Form Linear Program}
We conducted the same experiment as in \cref{sec:experiments}, solving the subgames using different subgame-solving techniques. However, we have used SQF as the solver rather than CFR. The results in \cref{fig:resolved_exploitability_lp} show that the results of using SQF as a solver are similar to those of CFR.
\begin{figure}[]
    \centering
        \begin{subfigure}[b]{0.24\textwidth}
            \centering
            \includegraphics[width=\textwidth]{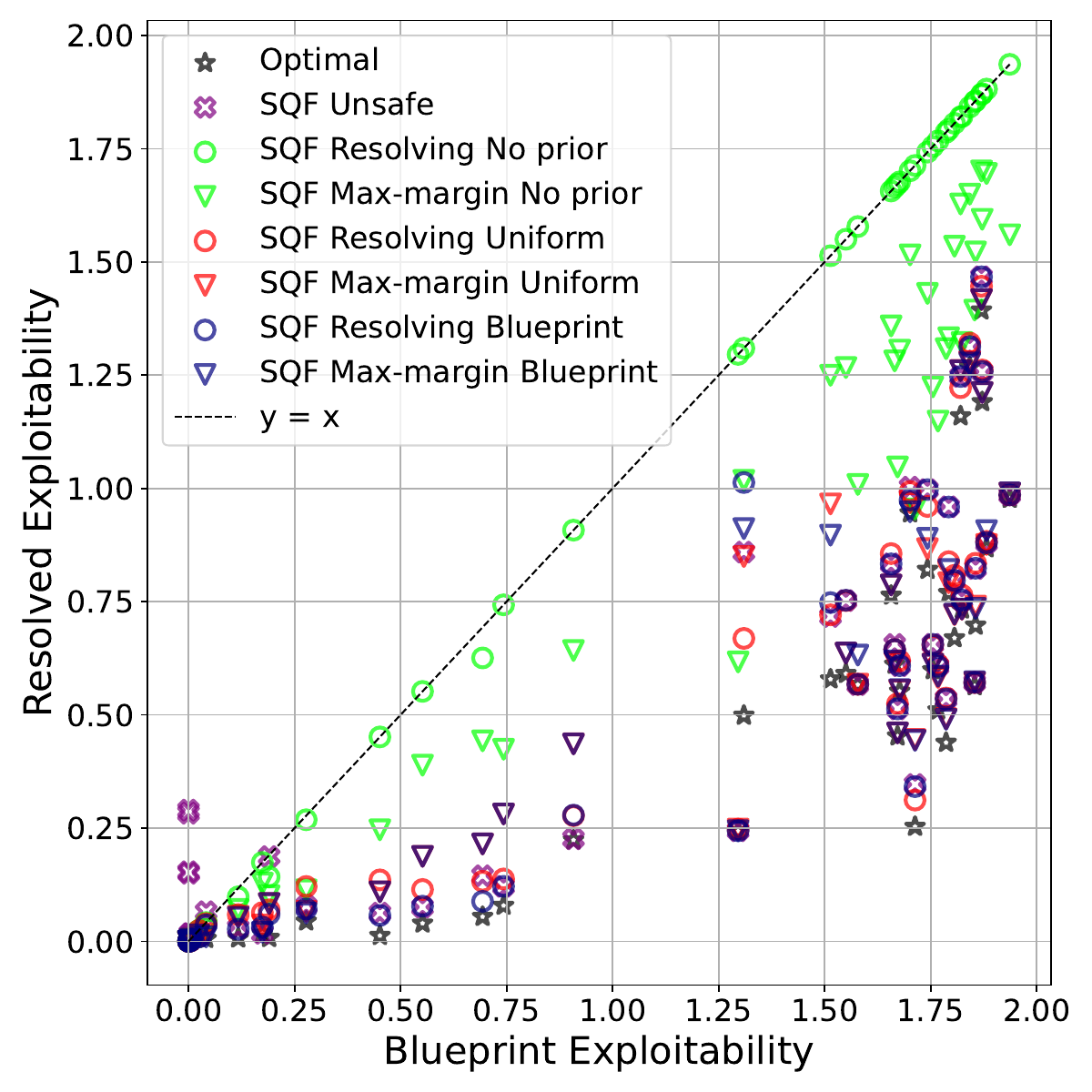} 
            \caption{Goofspiel 5, depth 1}
            \label{fig:expl_lp_goof5_d1}
        \end{subfigure}%
        \hfill
        \begin{subfigure}[b]{0.24\textwidth}
            \centering
            \includegraphics[width=\textwidth]{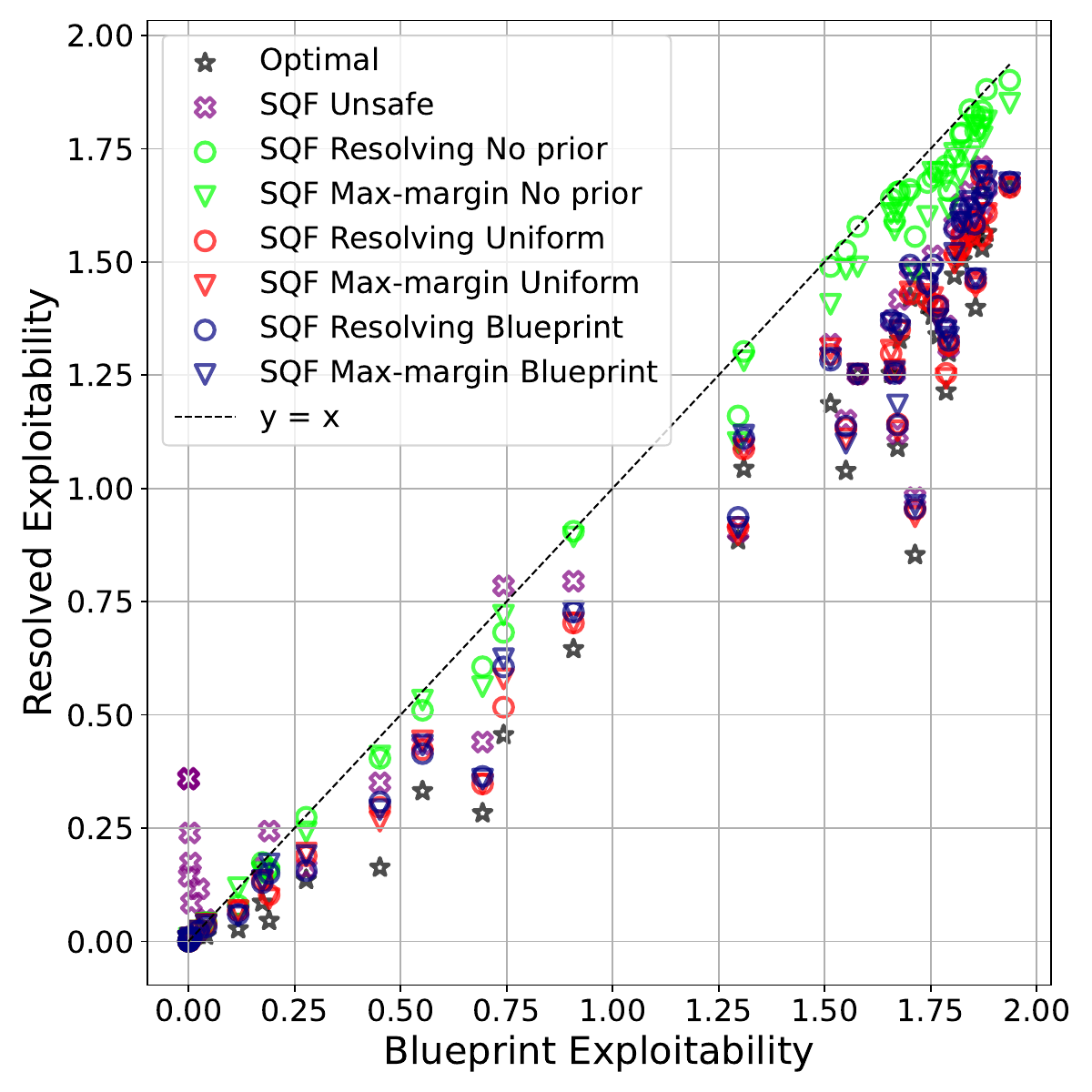} 
            \caption{Goofspiel 5, depth 2}
            \label{fig:expl_lp_goof5_d2}
        \end{subfigure}
        \begin{subfigure}[b]{0.24\textwidth}
            \centering
            \includegraphics[width=\textwidth]{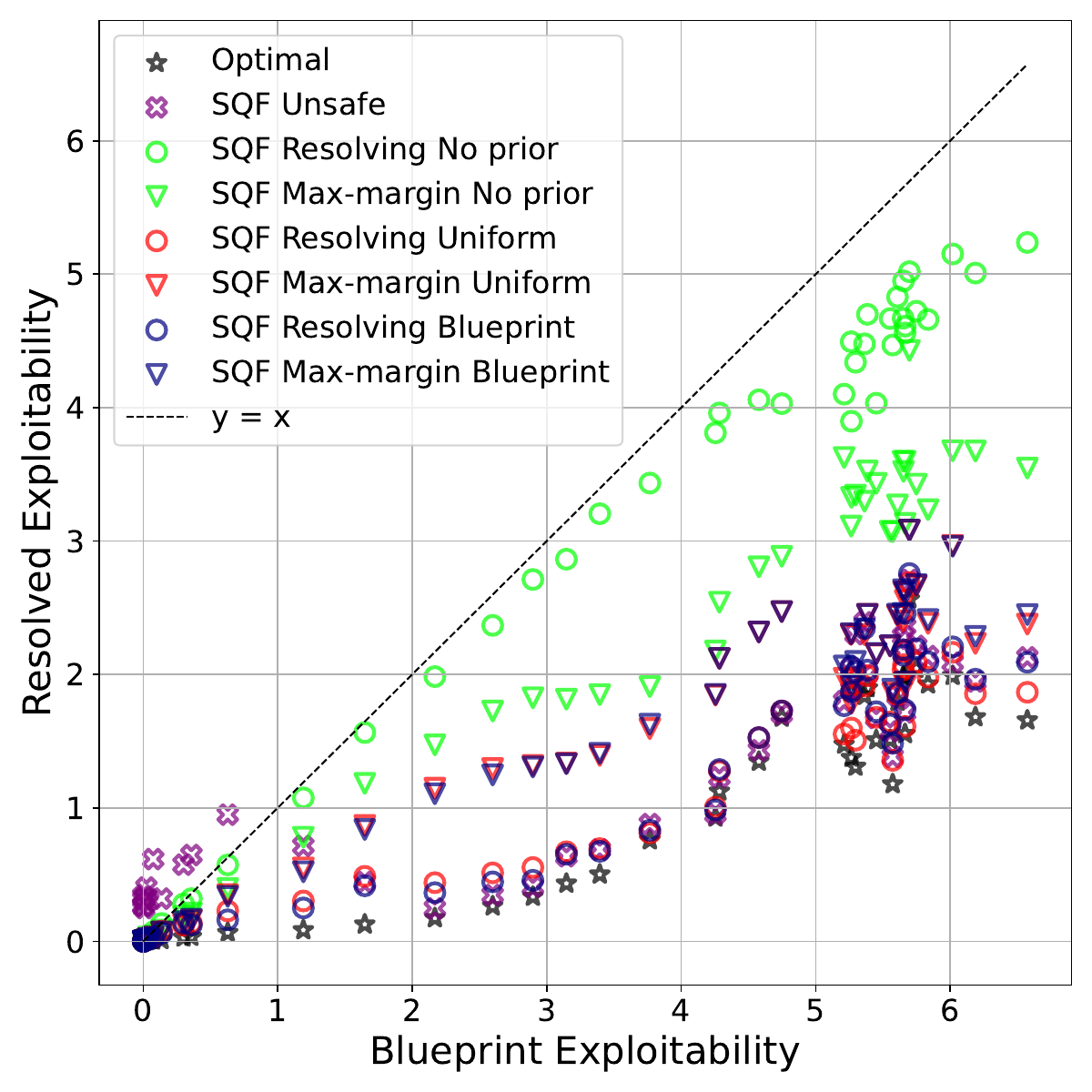} 
            \caption{Leduc hold'em after flop}
            \label{fig:expl_lp_leduc}
        \end{subfigure}%
        \hfill
        \begin{subfigure}[b]{0.24\textwidth}
            \centering
            \includegraphics[width=\textwidth]{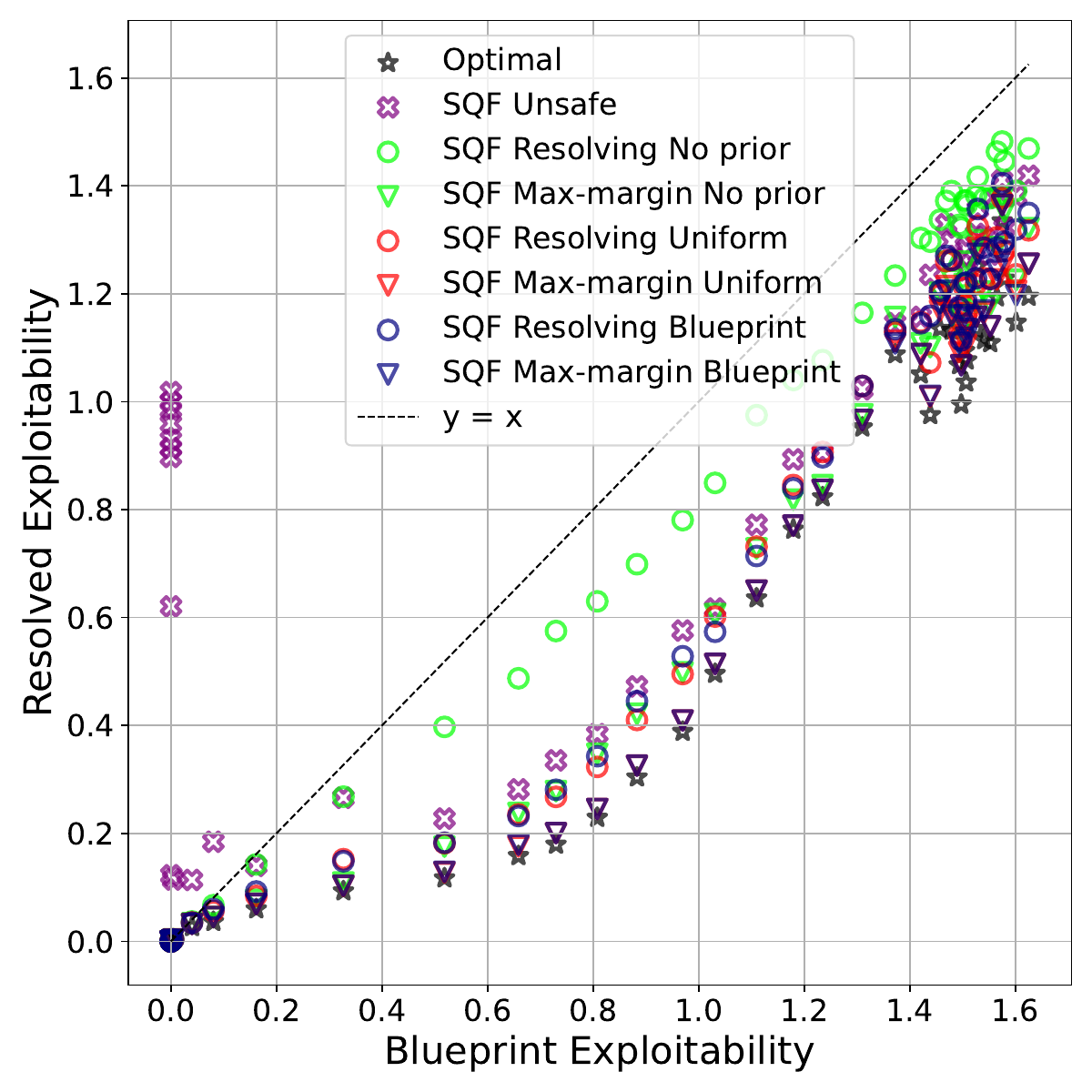} 
            \caption{Liar's dice, depth 1}
            \label{fig:expl_lp_liars_d1}
        \end{subfigure}
    \caption{Exploitability of the resolved strategy from each subgame in given depth based on blueprint's exploitability when using SQF as solver.}
    \label{fig:resolved_exploitability_lp}
\end{figure}
\subsection{Nested Subgame Solving}
In \cref{fig:heads_goof,fig:heads_liars} we provide the win-rates for each pair of subgame solving techniques, which were used in \cref{sec:nested}. The row technique corresponds to Player 1, while the column technique corresponds to Player 2. For the \cref{fig:heads}, we have averaged the results between Player 1 and Player 2 and used the minimum from each row in \cref{fig:heads_goof,fig:heads_liars}. In Liar's Dice, there is a small advantage to Player 2, so the value of the game is negative. However, we were not able to measure the exact value.

\begin{figure}
    \centering
    \includegraphics[width=\linewidth]{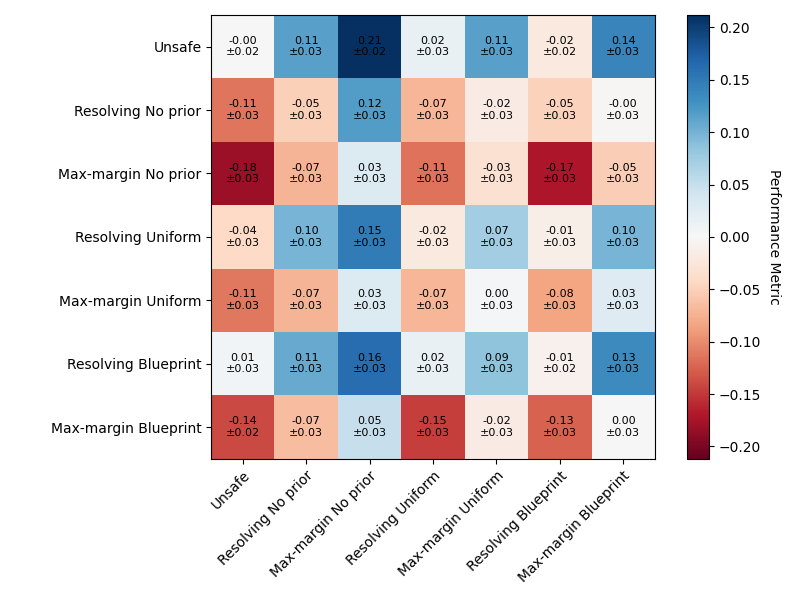}
    \caption{Head-to-head performance of different subgame solving techniques in Goofspiel 7. The win-rates are shown from the perspective of the row technique.}
    \label{fig:heads_goof}
\end{figure}
\begin{figure}
    \centering
    \includegraphics[width=\linewidth]{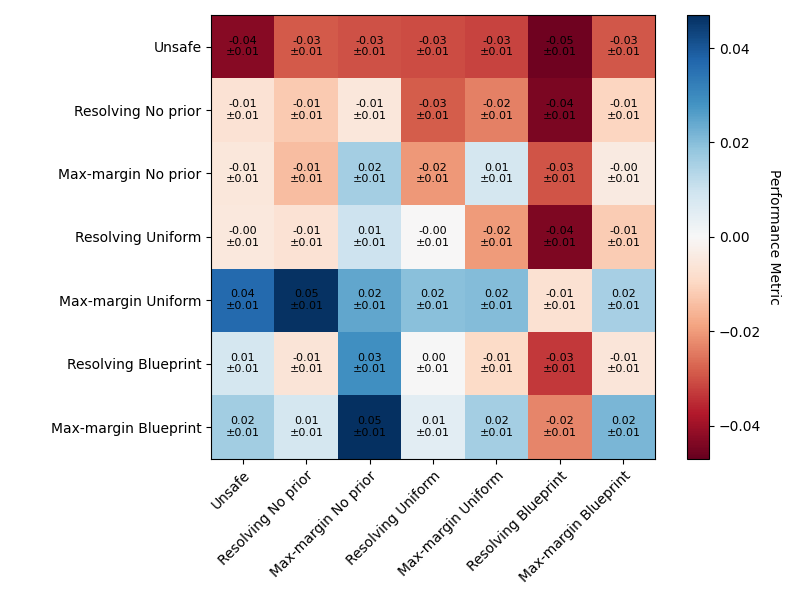}
    \caption{Head-to-head performance of different subgame solving techniques in Liar's Dice 2.4. The win-rates are shown from the perspective of the row technique.}
    \label{fig:heads_liars}
\end{figure}

\subsection{Hyperparameters and Other Details}
In \cref{tab:hyperparameters} we provide the hyperparameters used in the experiments. All experiments were run on a CPU cluster with AMD EPYC 7543, and each required at most 64 GB of RAM. For training the value functions in \cref{sec:nested}, we used a single Tesla A100 40GB.
\begin{table}[htbp]
    \centering
    \begin{tabular}{lll}
        \toprule \textbf{Experiment}& \textbf{Parameter} & \textbf{Value} \\
        \midrule
        \multirow{2}*{\parbox{2cm}{\textbf{Exploitability}}} 
            &CFR Iterations  & 5000 \\
            &Prior $\PriorEpsilon$ & $10^{-3}$\\
        \midrule
        \multirow{4}*{\parbox{2cm}{\textbf{Nested\\  solving}} } 
            &CFR Iterations  & 3000 \\
            & Prior $\PriorEpsilon$ & $2 \cdot 10^{-2}$\\
            & MVS strategies & 10 \\
            & Depth-limit & 1 \\
        \bottomrule
    \end{tabular}
    \caption{Hyperparameters.}
    \label{tab:hyperparameters}
\end{table}

%% file: chapters/game_rules.tex
\section{Game Rules}
We provide rules of all the games we have used in our experiments.

\subsection{Imperfect-Information Goofspiel}
Goofspiel N is played by two players, each starting with cards numbered 1 to N. The dealer has the same set of cards. Each turn, the dealer reveals a single card to both players, then both players secretly bet one of their cards. The dealer observes both cards and awards points to the player with the higher bid. In a case of the draw, the dealer discards the card. The players can play the same card at most once. In our experiments, the dealer's cards are dealt in a fixed order from the highest-valued card to the lowest.

\subsection{Leduc Hold'em}
Leduc hold'em is a toy variant of Texas hold'em poker. Only two players play the game with 6 cards in two suits: Spades, Hearts, and three ranks: Jack, Queen, King. The game starts with a mandatory bet of 1 coin from both players. Then each player is secretly dealt one card from the dealer. The game consists of two betting rounds, where each player can either fold, call, or raise. When the player folds, the game immediately ends, and the opponent receives all the coins in the pot. If the player calls, it adds coins to the pot, so there is the same number of coins from both players, and the game proceeds to the next betting round. If the player raises, it adds the same number of coins to the pot as in a call, and it also places either 2 more coins in the first round or 4 in the second round. In each turn, the players can raise at most twice. After the first round, the dealer deals a single public card, which both players observe. After the second round, the game ends, and the player with the higher combination of private and public cards wins and gets all the coins in the pot. The ranking of cards is that the player wins if they have the same card as the public one. If neither player has the same card as the public one, the higher-ranked card wins. The ranks are in this order from highest: King, Queen, Jack. If both players hold the same rank, the game ends in a draw, and the players split the pot.

\subsection{Liar's Dice}
Liar's dice N, M is played by two players, each holding N M-sided dice. The game starts by each player secretly rolling all of their dice. Then the players alternate in betting starting with player 1. The goal of the bet is to guess at least how many dice are rolled on a given side across all players, so each bet consists of two numbers: the side number and the number of dice rolled that side. Each subsequent bet has to either have a larger side number or more dice. Instead of a bet, a player can refute the opponent's bet. The game ends either when the largest bet is made or when the player rejects the bet. At the end, each player shows their dice, and if the bet was right, the player with that bet wins; otherwise, the refuting player wins. Moreover, the largest side number M serves as a universal number that can be treated like any other number.

%% file: chapters/sequence_form.tex
\section{Sequence-Form Linear Program}
We provide further details on the sequence-form linear program. $\SQFUtility{} \in \RealNumbers^{|\Sequences{1}| \times |\Sequences{2}|}$ is a utility matrix, which for each combination of sequences by both players contains a utility weighted by chance. $\SQFInitRealization{\Player} \in \RealNumbers^{|\Infoset{\Player}|}$ is a vector that contains the initial reaches of each information set. It contains 1 in the position corresponding to information sets, where the player $\Player$ plays for the first time, and 0 elsewhere. $\SQFRealizationConstraints{\Player} \RealNumbers^{|\Infoset{\Player}| \times \Sequences{\Player}} $. The sequence-form linear program differs in factored observation stochastic games, because each node is associated with an information set, compared to extensive-form games, where the information set is defined only for the acting player 
\begin{subequations}
\label{eq:sqf2}
\begin{align}
    \max_{\SQFStrategy{1}, \SQFValue{2}} &\SQFInitRealization{2}^\Transposition \SQFValue{2} \label{eq:sqf_obj2} \\
    \SQFRealizationConstraints{1} \SQFStrategy{1} &= \SQFInitRealization{1} \label{eq:sqf_primals2} \\
    \SQFUtility{}^\Transposition  \SQFStrategy{1}  - \SQFRealizationConstraints{2}^\Transposition \SQFValue{2}  &\geq \boldsymbol{0} \label{eq:sqf_duals2}  \\
    \SQFStrategy{1} &\geq \boldsymbol{0} \label{eq:sqf_probs2}
\end{align}
\end{subequations}

%% file: chapters/llm.tex
\section{LLM Usage}
Throughout the work on this paper, we have used several LLMs. The authors double-checked every LLM output before it was used in the final version. We have used these LLMs for the following tasks.
\begin{itemize}
    \item Gemini 3 was used during the writing process to improve the quality of the written text.
    \item Sonnet 4.5 was used during the implementation of experiments as a coding assistant.
\end{itemize}